\documentclass[a4paper,oneside,fleqn,11pt]{article}
\usepackage{amsmath,amssymb,amsfonts,amsthm,type1cm,bm,color,colortbl, mathrsfs}
\usepackage{graphicx,psfrag,epsf,booktabs, geometry}
\usepackage[dvipsnames]{xcolor}
\usepackage{natbib}
\bibpunct[:]{(}{)}{,}{a}{}{,}
\usepackage{rotating}
\usepackage{comment}
\usepackage{ulem}
\usepackage{float}
\usepackage{authblk}%
\usepackage[dvipdfmx]{}
\usepackage{accents}
\usepackage{yfonts}
\usepackage{pdflscape}
\usepackage{setspace}
\RequirePackage[colorlinks,citecolor=blue,urlcolor=blue]{hyperref}
\usepackage{bigstrut}
\usepackage{mathtools}
\usepackage{enumitem}
\mathtoolsset{showonlyrefs=true}

\DeclareMathOperator\E{\mathbb{E}}

\DeclareMathOperator\Cov{Cov}
\DeclareMathOperator\Pro{\mathbb{P}}

\DeclareMathOperator\FDR{FDR}
\DeclareMathOperator\Pwr{Power}

\def\cF{\mathcal{F}}
\def\cH{\mathcal{H}}

\def\bbR{\mathbb{R}}

\def\ep{\varepsilon}

\def\what{\widehat}

\newcommand{\tcr}{\textcolor{red}}

\newlength{\dhatheight}

\newtheorem{thm}{Theorem}
\newtheorem{lem}{Lemma}
\newtheorem{con}{Condition}
\newtheorem{rem}{Remark}
\newtheorem{defi}{Definition}
\newtheorem{ex}{Example}
\newtheorem{prop}{Proposition}
\newtheorem{proc}{Procedure}

\makeatletter
\def\section{\@startsection {section}{1}{\z@}{-3.5ex plus -1ex minus-.2ex}{2.3ex plus .2ex}{\large\bf}}
\makeatother

\makeatletter
\def\subsection{\@startsection {subsection}{1}{\z@}{-3.5ex plus -1ex minus-.2ex}{2.3ex plus .2ex}{\normalsize\bf}}
\makeatother

\title{\textbf{Robust Reproducible Network Exploration}}

\author[$*$]{\textsc{Masaki Toyoda}
}
\author[$\dagger$]{\textsc{Yoshimasa Uematsu}
}

\affil[*]{\textit{Department of Economics, Hitotsubashi University}}
\affil[$\dagger$]{\textit{Department of Social Data Science, Hitotsubashi University}}

\begin{document}

\renewcommand{\theequation}{\thesection.\arabic{equation}}
\makeatletter
\@addtoreset{equation}{section}
\makeatother

\maketitle

\begin{abstract}
We propose a novel methodology for discovering the presence of relationships realized as binary time series  between variables in high dimension. To make it visually intuitive, we regard the existence of a relationship as an edge connection, and call a collection of such edges a network. Our objective is thus rephrased as uncovering the network by selecting relevant edges, referred to as the network exploration. Our methodology is based on multiple testing for the presence or absence of each edge, designed to ensure statistical reproducibility via controlling the false discovery rate (FDR). In particular, we carefully construct $p$-variables, and apply the Benjamini-Hochberg (BH) procedure. We show that the BH with our $p$-variables controls the FDR under arbitrary dependence structure with any sample size and dimension, and has asymptotic power one under mild conditions. The validity is also confirmed by simulations and a real data example.
\end{abstract}
\textbf{Keywords.} Exploratory data analysis, Multiple test, False discovery rate control, Power, E-value.

\section{Introduction}

Science advances through the iterative process of hypothesis formulation and confirmation. In this process, traditional statistics has played a crucial role in \textit{confirmatory data analysis} (the upper box in Figure \ref{fig:DS}). This involves collecting data according to hypotheses or models derived from theoretical frameworks of interest, and then testing the validity of those hypotheses or models against the data. The empirical results are then fed back into the theories and hypotheses. However, with the recent accumulation of big data, the role of statistics has expanded to include the formulation of hypotheses (the lower box in Figure \ref{fig:DS}). In particular, the importance of \textit{exploratory data analysis} \citep{tukey1962future} for discovering key variables and relationships hidden in high-dimensional big data has increased. The variables and hypotheses discovered to be important contribute directly (Route \textsf{A} in Figure \ref{fig:DS}), or after further confirmatory data analysis (Route \textsf{B} in Figure \ref{fig:DS}), to the construction of new theories. This data-driven theory construction, often referred to as \textit{data science}, aligns with the views of \cite{donoho201750} and \cite{nosek2018preregistration}.

In this article, we focus on exploratory data analysis, particularly proposing a methodology for discovering the presence of ``relationships'' between variables within large datasets. In our framework, any relationship can be considered without any assumption as long as its existence is realized as binary time series. For example, in SNS data, a relationship between a pair of users exists if they are likely to send messages each other over a period. In another example of traffic data, a relationship between two intersections exists if the traffic volume between them exceeds a certain threshold. For the sake of visual clarity, we regard the existence of a relationship between two variables as an edge connecting the two nodes, and call a collection of such edges a network. 
Our objective is thus rephrased as stably uncovering the network hidden within the big data by selecting relevant edges. We call this the \textit{network exploration}. Although not within the scope of this paper, the discovered networks may be subject to further confirmatory analysis as necessary.

\begin{figure}[!t]
\centering
\includegraphics[width=13cm]{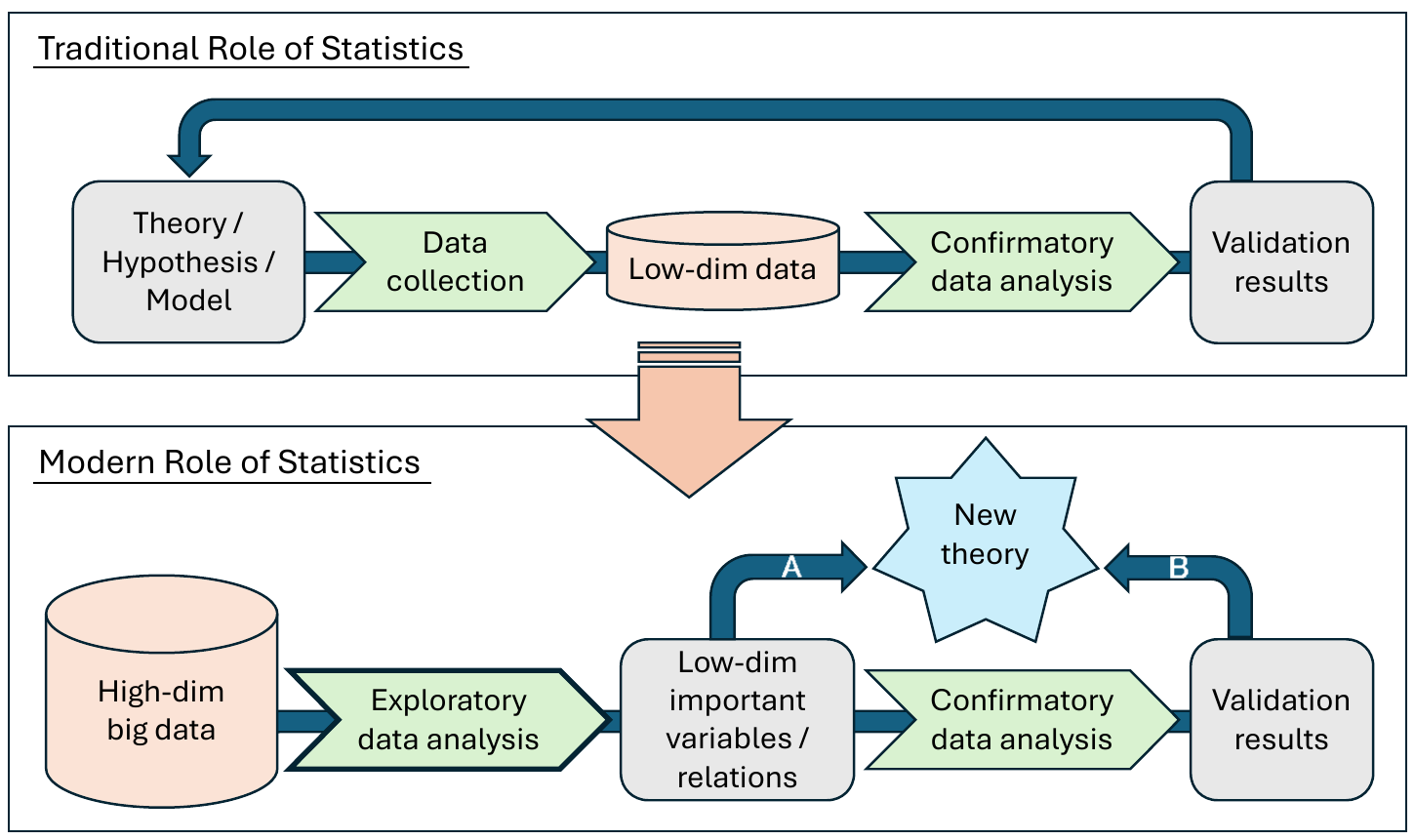}
\caption{The evolution of the Role of Statistics after big data. Our focus is on the exploratory data analysis to determine whether important ``relationships'' exist between variables in high-dimensional data. The discoveries can contribute directly (Route \textsf{A}), or after further confirmatory data analysis (Route \textsf{B}), to building new theories, if they are reproducible. }\label{fig:DS}
\end{figure}

Our methodology for network exploration is based on multiple testing for the presence or absence of each edge. Importantly, our framework is designed to ensure statistical reproducibility. While exploratory data analysis may not always require reproducibility, especially when followed by a more refined confirmatory analysis (Route \textsf{B} in Figure \ref{fig:DS}), reproducibility becomes essential when its results are used directly to construct new theories (Route \textsf{A} in Figure \ref{fig:DS}). 
A crucial aspect of our methodology for reproducible network exploration is to control the \textit{false discovery rate} (FDR) of edge discovery, which will contribute stability in network detection. 
This together with high power for selecting significant edges ensures reproducibility in network detection, thereby providing a foundation for theory building. In particular, we expect it to be valuable in fields where experimentation is challenging, such as economics and other social sciences.

We adopt the method of \cite{benjamini1995controlling}, called the BH, for edge selection in pursuit of the FDR control. Whether this works well or not depends on the construction of $p$-variables for each hypothesis. In general, the BH does not guarantee the FDR control unless the $p$-variables satisfy the property, \textit{positive regression dependence on the subset of nulls} (PRDS); see \cite{benjamini2001control, FinnerEtAl2009, RamdasEtAl2019unified}. However, the condition is difficult to verify for a given set of $p$-variables. As our methodological contribution, we construct two kinds of $p$-variables for a robust discovery against such dependencies. The first one is based on $e$-variables, which are non-negative random variables taking their expectations less than or equal to one under the null \citep{vovk2021values}. It is known that the reciprocals of $e$-variables become $p$-variables. The BH that uses such $p$-variables is called the e-BH, and is known to be robust against any dependence structure. Our $e$-based $p$-variable can utilize domain knowledge or beliefs, leading to power enhancement. The second construction is based on misspecification. Our null hypotheses of the edge detection problem are complicated enough not to directly form $p$-variables. We then deliberately misspecify each null to be the null closest to the alternatives, which still yields proper $p$-variables. Although they have rather complicated dependence, \cite{benjamini2001control} procedure, called the BY, helps perform robust inference.

In terms of statistical theory, we prove that the BH with our $e$-based $p$-variables and the BY with our misspecification-based $p$-variables control the FDR of edge discoveries less than or equal to a preassigned level even under arbitrary dependence structure. Moreover, under a set of mild conditions, we also show that their power tends to unity. These theoretical results will bring statistical reproducibility in the network exploration, and hence the obtained network can directly be a basis for a new theory via Route \textsf{A}.   We demonstrate the performances of the BH with our $p$-variables by numerical simulations and empirical example with daily SNS data. 

Our study appears to offer a unique level of novelty, as we have not found papers that take a similar approach to exploring the presence or absence of relationships (i.e., network exploration) with FDR control. That being said, we will highlight some related works. We can find papers considering estimation of network models, including TERGM \citep{hanneke2010discrete, krivitsky2014separable}, rewiring Markov chain \citep{crane2015time}, and supposing that similar nodes are likely to be connected \citep{liben2003link, pandey2019comprehensive}. To model static networks, independently drawn edges are often assumed; this category includes the stochastic block models  \citep{holland1983stochastic} and its variants \citep{airoldi2008mixed, karrer2011stochastic}. There are also works aiming to detect some specific relationships defined as the covariance or higher moments \citep{de2004discovery, faith2007large, bien2011sparse, yu2016sparse}. 

The remaining part of this paper is organized as follows. 
Section \ref{sec:preliminaries} introduces the formal definition of edges and networks, and formulates the edge detection problem as a multiple testing. Section \ref{sec:Methodology} proposes our methodology for detecting edges. Specifically, we propose two sorts of $p$-variables that are used in the BH (or BY). Section \ref{sec:Theory} investigates the theoretical properties of the proposed method. First we claim that our $p$-variables are reasonable, and then prove the FDR control and power guarantee. Section \ref{sec:Simulations} demonstrates through simulations that the proposed method performs well in terms of both FDR and power. Section \ref{sec:Real data} applies our method to real SNS data to discover a hidden network. Section \ref{sec:Conclusion} is a conclusion. All the proofs of our theoretical results are collected in Appendix.

\section{Preliminaries}\label{sec:preliminaries}

We formulate the problem of network detection within the framework of multiple testing. 
Section \ref{subsec:network} introduces the $\pi$-connectable network, which is our target to discover. 
Section \ref{sec:mt} formulates our network detection problem as multiple testing, and reviews the statistical measures called the FDR and power. 

\subsection{Connectable networks}\label{subsec:network}

We begin with a formal definition of edges that comprise a network. As described in Introduction, we represent a presence of relationship between two variables as a connection of edge between two nodes, and whether it connects or not will be detected by a binary time series. 
More precisely, we define the $i$th \textit{potential edge} (or simply \textit{edge}) at time $t$ as a Bernoulli random variable $X_{it}$ indexed by $i\in[n]:=\{1,\dots,n\}$ and $t\in[T]$ with the parameter given by a random variable $\xi_{it}=\E[X_{it}\mid \cF_{t-1}] \in [0,1]$ a.s., where $\cF_t$ is a $\sigma$-field generated by $\{X_{is}:i\in[n],\ s\leq t\}$. We say that the $i$th edge is \textit{connected} at time $t$ if $X_{it}=1$. In our setting, we do not explicitly define what the nodes are, but instead directly index the edges. The parameter $\xi_{it}$ that reflects the past information $\cF_{t-1}$ controls the possibility of edge appearance. Remarkably, we do not require any further assumption on the dependence structures across $i\in[n]$ and $t\in[T]$, which is desirable in exploratory data analysis.


We next define edges that are likely to be connected in a certain period. An $i$th sequence of  potential edges, $\{X_{it}:t\in[T]\}$, is said to be \textit{$\pi$-connectable} up to time $T$ if the corresponding parameter sequence  $\{\xi_{it}:t\in[T]\}$ satisfies
\begin{align}
\Pro\left(\max_{t\in[T]}\xi_{it}>\pi\right)>0
\end{align}
for given $\pi\in[0,1]$. Whether the $i$th sequence of edges is $\pi$-connectable or not depends on the value of $\pi$. 
Our target of exploration, the \textit{$\pi$-connectable network} $\cH_1(\pi)$, is defined as a set of $\pi$-connectable edges up to time $T$: 
\begin{align}
\cH_1(\pi)=\left\{i\in[n]:\Pro\left(\max_{t\in[T]}\xi_{it}>\pi\right)>0\right\}.
\end{align}
In other words, $\cH_1(\pi)$ is a set of edges that are likely to be connected at least in a certain period over $[T]$. Thus it does not mean that all the members in $\cH_1(\pi)$ are always connectable over $[T]$ and beyond. Thanks to this character, our networks to be detected are interpretable even if the time series of potential edges contains structural breaks, such as a level shift and/or slope change, in $[T]$. This actually fits our goal: we are willing to do exploratory research for detecting edges worth investigating in detail. The detected edges may be investigated more in further studies. On the contrary, most of the existing works on networks aim to conduct estimation of some parametric models that are time invariant; see \cite{hanneke2010discrete} and \cite{krivitsky2014separable}, for instance. They essentially need stationarity in data. 

The predetermined parameter $\pi$ can be interpreted as the sparsity parameter in the network because $\cH_1(\pi)$ becomes sparser/denser as $\pi$ gets closer to one/zero. For example, we have $\Pro(\max_{t\in[T]}\xi_{it}>1)=0$ for all $i\in[n]$ by the definition of $\xi_{it}$, meaning that the $1$-connectable network is the empty (sparsest) set. 
A similar formulation of $\cH_1(\pi)$ can be found in, e.g., \cite{henzi2022valid} and \cite{fan2023testing}.

Our modeling framework is general enough to include many networks that have dependence across $t$ and/or $i$. We give some examples as follows. 

\begin{ex}[Independent Bernoulli]\normalfont
Suppose independence of $X_{it}$ across $t$ for each $i\in[n]$ and $\xi_{it}=\pi_i$ a.s.\ for some constant $\pi_i\in(0,1)$. Then, for each $i\in[n]$, we have $X_{it}\sim \text{i.i.d.}\ Ber(\pi_i)$, and the $\pi$-connectable network reduces to $\cH_1(\pi)=\{i\in[n]:\pi_i>\pi\}$. 
\end{ex}

\begin{ex}[Bernoulli VAR]\normalfont
Suppose that $Y_t=(Y_{1t},\dots,Y_{nt})^\top$ and $Z_t=(Z_{1t},\dots,Z_{nt})^\top$ are independently sampled from a multivariate Bernoulli distribution with joint probabilities $\pi^y=(\pi_1^y,\dots,\pi_n^y)^\top$ and $\pi^z=(\pi_1^z,\dots,\pi_n^z)^\top$. 
An $n$-dimensional binary vector series $X_t=(X_{1t},\dots,X_{nt})^\top$ is generated by the Benoulli VAR model \citep{EuanSun2020} if 
\begin{align}
    X_{it}=Y_{it}X_{i,t-1}+Z_{it}(1-X_{i,t-1}),~~~i\in[n].
\end{align}
If $\{X_t\}$ is weakly stationary, then the $\pi$-connectable network reduces to $\cH_1(\pi)=\{i\in[n]:\pi_i>\pi\}$, where $\pi_i:=\E[X_{it}]=\pi_i^z/(1-\pi_i^y+\pi_i^z)$. 
\end{ex}

\begin{ex}[Logistic model]\normalfont
Suppose that an $n$-dimensional binary vector series $X_t=(X_{1t},\dots,X_{nt})^\top$ follows the logistic model:
\begin{align}
    \xi_{it}^{\beta}:=\E[X_{it}\mid X_{t-1}]=\frac{1}{1+\exp\left\{-(\beta_{i0}+\beta_{i}^\top X_{t-1})\right\}},
\end{align}
where $(\beta_{i0},\beta_{i}^\top)^\top\in\bbR^{1+n}$ is a coefficient vector. A similar model can be found in \cite{azzalini1994logistic}. 
In this case, the $\pi$-connectable network reduces to $\cH_1(\pi)=\{i\in[n]:\Pro(\max_{t\in[T]}\xi_{it}^{\beta}>\pi)>0\}$.
\end{ex}

\subsection{FDR-controlled network exploration via multiple test}\label{sec:mt}

We formulate the network exploration as a multiple test. For given $\pi\in[0,1]$, define the set of indices of null edges by 
\begin{align}\label{H0}
\cH_0(\pi) 
= [n] \setminus \cH_1(\pi) 
= \left\{i\in[n]:\Pro\left(\max_{t\in[T]}\xi_{it}\leq \pi\right)=1\right\},
\end{align}
leading to a sequence of hypotheses:
\begin{align}\label{eq:test}
H_{0i}:i\in\cH_0(\pi)~~~\text{versus}~~~H_{1i}:i\in\cH_1(\pi)~~~\text{for each }i\in[n].
\end{align}
A multiple test for \eqref{eq:test} by some statistical procedure gives a set of rejected nulls (i.e., discovered alternatives), $\what{\cH}_1(\pi)$, which is the detected $\pi$-connectable network. The detail on the procedure is explained in Section \ref{ssec:BH}. 

Through this network exploration in \eqref{eq:test}, we aim to gain insights that will advance the development of new network theories. For this purpose, the discovered network should be statistically reproducible. Therefore, we are considering the type I error control of this multiple test by the suppression of falsely detected edges. 
The \textit{false discovery rate} (FDR), a measure of the type I error, of $\what{\cH}_1(\pi)$ for \eqref{eq:test} is defined as 
\begin{align}
\FDR(\pi)=\E\left[\frac{|\cH_0(\pi)\cap\what{\cH}_1(\pi)|}{|\what{\cH}_1(\pi)|}\right].
\end{align}
We attempt to control the FDR to be less than or equal to a pre-assigned level $\alpha\in(0,1)$. At the same time, it is desirable to configure $\what\cH_1(\pi)$ to include the members of $\cH_1(\pi)$ as much as possible. Specifically, we aim for $\what{\cH}_1(\pi)$ to have high power. Here, \textit{power} refers to the measure of avoiding a Type II error in multiple tests, and is defined as 
\begin{align}
\Pwr(\pi)=\E\left[\frac{|\cH_1(\pi)\cap\what{\cH}_1(\pi)|}{|\cH_1(\pi)|}\right].
\end{align}
A powerful FDR-controlled procedure contributes to a stable reproducible network exploration. We should note that both $\FDR(\pi)$ and $\Pwr(\pi)$ depend on $n$ and $T$, but we omit writing them explicitly.

\section{Methodology}\label{sec:Methodology}

To conduct the multiple testing for \eqref{eq:test}, we apply the procedure of \cite{benjamini1995controlling} (BH hereafter), which requires $p$-variables. In Section \ref{ssec:BH}, we formally define $p$-variables, and review the BH procedure. Section \ref{sec:e} contains our main methodological contributions; Sections \ref{subsec:e} presents valid $p$-variables via $e$-variables that can address some challenges discussed later, and Section \ref{subseccomparison} proposes another $p$-variable based on misspecification.

\subsection{The BH procedure}\label{ssec:BH}

Before reviewing the BH procedure for discovering the $\pi$-connectable network $\cH_1(\pi)$, we formally define $p$-variables since they play a crucial role in the procedure. A nonnegative random variable $P$ is called a \textit{$p$-variable} if $\Pro\left(P\leq\alpha\right)\leq\alpha$ holds for all $\alpha\in(0,1)$ under the null hypothesis. Its realization $p$ is called a \textit{$p$-value}.

We next review the BH procedure. Given target level $\alpha\in(0,1)$ and a set of $p$-variables $\{P_i\}_{i=1}^n$ along with the nulls in \eqref{eq:test}, define 
\begin{align}
K=\max\left\{i\in[n]:P_{(i)}\leq\frac{i\alpha}{n}\right\},
\end{align}
where $P_{(i)}$ is the $i$th smallest $p$-variable in $\{P_i\}_{i=1}^n$. Then the BH detects the $\pi$-connectable network by 
\begin{align}
\what\cH_1(\pi)=
\begin{cases}
\left\{ i\in[n]:P_{i}\leq P_{(K)} \right\} & \text{if $K$ exists}; \\
\emptyset & \text{otherwise}.
\end{cases}
\end{align}
In the BH procedure, construction of $p$-variables is important. As is well known in the literature, the BH procedure ensures proper FDR control when the $p$-variables are either independent or positively correlated, the latter condition being more precisely characterized by PRDS; see Introduction for references. However, it seems quite difficult to check the property in real data applications, including the network detection, due to the existence of complex dependence structures among edges. In some applications, it might be challenging to even construct $p$-variables in a usual way if the null hypothesis is complicated.

\subsection{Constructions of p-variables}\label{sec:e}

As our main methodological contribution, we propose two strategies for construction of $p$-variables; the first one is based on $e$-variables and the second is on misspecification.

\subsubsection{Robust e-based p-variables}\label{subsec:e}

The first strategy utilizes $e$-variables \citep{vovk2021values}; a nonnegative random variable $E$ is called an \textit{$e$-variable} if $\E[E]\leq1$ holds under the null hypothesis. Its realization $e$ is called an \textit{$e$-value}. 
Importantly, $1/E$ becomes a $p$-variable under the null, which is easily verified by the Markov inequality: 
$\Pro\left(1/E\leq\alpha\right)=\Pro\left(E\geq1/\alpha\right)\leq\alpha\E[E]\leq\alpha$, satisfying the definition of $p$-variables. 
Thus we can use $P_i=1/E_i$ in the BH, which is called the e-BH \citep{wang2022false}. This procedure controls the FDR under any dependence structure among $e$-variables. Because the definition of $e$-variables cares not about the tail probability but about the first moment, we may construct $e$-variables for even such a complicated hypothesis. 

Let $\{\cF_t\}_{t=1}^T$ denote a filtration, where $\cF_t$ is a $\sigma$-field generated by $\{X_{is}:i\in[n], s\leq t\}$. Then, our robust $p$-variables are constructed via $e$-variables as follows. 
\begin{proc}\normalfont\label{proc:e}
For each $i\in[n]$, perform the following steps:
\begin{enumerate}[label=Step \arabic*., leftmargin=*]
\item Define the sequence of $e$-variables as 
\begin{align}
E_{it}&=\prod_{s=1}^{t}\left(\lambda_{is} X_{is}+1-\pi\lambda_{is}\right), ~~~ t\in[T],\label{Eit}
\end{align}
where $\lambda_{is}$ is any $[0,1/\pi)$-valued $\cF_{s-1}$-measurable random variable for $s\in[T]$.
\item For the stopping time 
\begin{align}
\tau_{i}=
\begin{cases}
\inf\left\{t\in[T]:E_{it}\geq n/\alpha \right\} & \text{if $\{t\in[T]:E_{it}\geq n/\alpha\}$ is not empty}; \\
\infty & \text{otherwise},
\end{cases}
\end{align}
define the stopped process as $E_i = E_{i,T\land\tau_{i}}$, which forms an $e$-variable. 
\end{enumerate}
Then our robust $e$-based $p$-variables are obtained by $P_i^{\normalfont{\textsf{e}}}=1/E_i$ for $i\in[n]$. 
\end{proc}

In Step 1, by the definition, $E_{it}$ becomes large as $X_{is}=1$ occurs more for $s\in[t]$. Thus a large value of $E_{it}$ is an evidence against $H_{0i}$ at time $t$. Although the form of $E_{it}$ in \eqref{Eit} is a bit complex, a similar construction can be found in \cite{waudby2023estimating}. Moreover, we will verify in Section \ref{constructionofstatistic} that any ``valid'' $e$-variable can be represented in this form. This $e$-variable contains hyper-parameter $\lambda_{it}$, which is a user-defined random variable that can reflect users' domain knowledge or belief. We call this property the data-adaptive feature. Desirably, if $\lambda_{it}$ is designed to take large values for $t\in[T]$ when $\xi_{it}$ is large, power enhancement will be expected. 
We will see  an example for such a construction in Section \ref{sec:periodic}. 

Step 2 stipulates the optional stopping rule of the multiplication in \eqref{Eit}; it is stopped as soon as the evidence against the null is sufficiently accumulated, and we can reject $H_{0i}$ without looking at what comes after. This stopping property is suitable for our hypotheses \eqref{eq:test}, which is constructed with a maximum of $\xi_{it}$ over $[T]$. As mentioned before, our objective is to detect edges that are likely to be connected at least once rather than on average. 
For example, suppose that we have a sequence of data with a strong signal against the null only in the beginning of $[T]$. If we use a sample average over the whole sequence to calculate a $p$-value, then it may result in a large value because of the weak signals in the later part of the sequence. In this case, we want to stop looking at data as soon as sufficient evidence accumulates. Unfortunately, this optional stopping rule is prohibited for ``conventional'' $p$-variables; we have to make a decision after all the data are observed. On the other hand, as \cite{grunwald2020safe} verifies, our $e$-based $p$-variable does not need to wait until observing all. We will see an example for this in Section \ref{level}.

\subsubsection{Misspecification-based p-variables}\label{subseccomparison}


As the second strategy, we intentionally misspecify the model under the null. To see this, consider a single test for an arbitrarily fixed $i$th pair of hypotheses in \eqref{eq:test} with given $\pi\in[0,1]$. Then, to construct the test statistic, we misspecify the model to assume $\xi_{it}=\pi$ a.s.\ for all $t\in[T]$ as the parameter under the null, regardless of the original null distribution of $X_{it}$. This is possible because it is the null closest to the alternative; in other words, any null hypothesis is more conservative than $\xi_{it}=\pi$ a.s. The crucial point is that this misspecified assumption is equivalent to assuming $X_{it}\sim \text{i.i.d.}\ Ber(\pi)$. Therefore, we can construct the $p$-variable by a standard argument as follows: 
\begin{align}\label{p-value}
P_{i}^{\normalfont{\textsf{m}}}
= \sum_{t=S_i}^{T}
\begin{pmatrix}
T\\t
\end{pmatrix}
\pi^t(1-\pi)^{T-t},
\end{align}
where $S_i = \sum_{t=1}^{T}X_{it}$. We call this the misspecification-based ($m$-based) $p$-variable. This construction of $p$-variables for \eqref{eq:test} is also new to the literature. If $X_{it}\sim \text{i.i.d.}\ Ber(\pi)$ is actually true, a test using this $p$-variable is known to be uniformly most powerful; see Theorem 3.4.1 and Example 3.4.2 in \cite{LehmannRomano2005TSH}. We will verify that $P_{i}^{\normalfont{\textsf{m}}}$ is indeed a proper $p$-variable under our null in Theorem \ref{misspecifiedp-value} in the next section. 

As rigorously seen later, both $P_i^{\normalfont{\textsf{e}}}$ and $P_i^{\normalfont{\textsf{m}}}$ are indeed $p$-variables under our setting in Section \ref{sec:preliminaries}, but unlike the former, using the latter in the BH procedure may not guarantee the FDR control due to a complicated dependency in $\{P_i^{\normalfont{\textsf{m}}}\}$. 
As a remedy, applying the BH procedure to $\{P_{i}^\text{m} h_n\}$ with $h_n = \sum_{k=1}^{n}1/k$, which are again $p$-variables, is an alternative choice. This procedure will lead to a FDR-controlled selection because this coincides with the procedure by \cite{benjamini2001control}. We call the BH procedure with $\{P_{i}^\text{m} h_n\}$ the BY. Intuitively, the $p$-variable $P_{i}^\text{m} h_n$ is large enough to allow any dependence structure among $\{P_{i}^\text{m}\}$ in the BH procedure.

\subsection{Comparison}

So far, we have proposed two kinds of $p$-variables, which give three  procedures to detect the $\pi$-connectable network: the BH with $\{P_i^m\}$, BY (i.e., BH with $\{P_i^mh_n\}$), and e-BH (i.e., BH with $\{P_i^e\}$). Table \ref{comparison} summarizes the characteristics of each procedure. In terms of robustness against unknown dependency in the underlying network, we should use the BY or e-BH. Again, the BH with $\{P_i^m\}$ should not be used since it could fail in controlling the FDR due to the unknown dependence. 

In comparison to the BY, the e-BH has the special properties: data adaptability and stopping rule. The e-BH can utilize information researchers possess by reflecting it to the construction of $\lambda_{it}$. Moreover, the e-BH allows us to stop updating the $e$-variable as soon as sufficient evidence accumulates. In contrast, the BY requires the use of all the data for computing $P_{i}^m$. As observed in Section \ref{sec:Simulations}, however, $P_i^e$ is in general less powerful. It might be recommended to use $P_i^mh_n$ rather than $P_i^e$ unless we have knowledge or belief on the underlying network. For instance, in SNS network, it is known that there is a clear separation between periods of being connected and not. If connective period comes earlier, we can expect $P_i^e$ to be powerful thanks to the stopping rule. Indeed, we observe the property in Section \ref{sec:Real data}.

\begin{table}[h]\centering
\caption{Comparison of $p$-variables in the BH procedure}\label{comparison}
\begin{tabular}{l|ccc}
\hline
& arbitrary dependence & data adaptability & stopping rule \\ 
\hline
BH with $\{P_i^{\normalfont{\textsf{m}}}\}$  & --  & --  & --              \\
BY (BH with $\{P_i^{\normalfont{\textsf{m}}}h_n\}$)   & $\checkmark$  & --  & --  \\
e-BH (BH with $\{P_i^{\normalfont{\textsf{e}}}\}$)  & $\checkmark$   & $\checkmark$   & $\checkmark$  \\
\hline
\end{tabular}
\end{table}


\section{Statistical Theory}\label{sec:Theory}


In Section \ref{constructionofstatistic}, we give a theoretical justification for the construction of $P_i^e$ and $P_i^m$. In Section \ref{fdrcontrol}, we will show that the e-BH and BY procedures in Section \ref{sec:Methodology} control the FDR. In section \ref{powerguarantee}, we will discuss when the procedures achieve asymptotic power tending to one.

\subsection{Theory for our p-variables}\label{constructionofstatistic}


We first give a justification for $\lambda_{it} X_{it}+1-\pi\lambda_{it}$ in \eqref{Eit}, where $\lambda_{it}$ is any $[0,1/\pi)$-valued $\cF_{t-1}$-measurable random variable for $t=1,\dots,T$. 
Denote by $\{e( \, \cdot \, ;\cF_{t-1})\}_{t=1}^T$ a sequence of random functions defined on $\{0,1\}$ such that $e(x ;\cF_{t-1})$ is an $\cF_{t-1}$-measurable random variable for all nonrandom $x\in\mathbb{R}$ and $e(X_{it} ;\cF_{t-1})$ becomes an $e$-variable for any $\pi\in(0,1)$ and $(i,t)\in\cH_0(\pi)\times[T]$. 
Hereafter, we write $e_{t-1}(\, \cdot \,)=e( \, \cdot \, ;\cF_{t-1})$ for simplicity. 

\begin{defi}\normalfont \label{def:valid}
For any $\pi\in(0,1)$ and $i\in\cH_0(\pi)$, a sequence of $e$-variables $\{e_{t-1}(X_{it})\}_{t=1}^T$ satisfying Conditions (A)--(C) a.s.\ is said to be \textit{valid}: 
\begin{enumerate}
\item[(A)] $0<e_{t-1}(0)\leq e_{t-1}(1)$ for any $t\in[T]$;
\item[(B)] $\E[e_{t-1}(X_{it})\mid\cF_{t-1}]\leq1$ for any $(i,t)\in\cH_0(\pi)\times[T]$;
\item[(C)] $\E[e_{t-1}(X_{it})\mid\cF_{t-1}]=1$ for any $(i,t)\in\{\xi_{it}=\pi\}$.
\end{enumerate}
\end{defi}
Condition (A) requires valid $e$-variables to be positive. Thus their products are also positive. Thanks to the second inequality in (A), we can interpret a large $e$-value as the evidence against the null hypothesis. Condition (B) is needed for valid $e$-variables to have the optional stopping property (which is used in the theorems we will introduce). For Condition (C), the event $\{\xi_{it}=\pi\}$ means that $i$th edge is a member of $\cH_0(\pi)$ and is on the boundary between the null and non-null. For such $i$th edge, the conditional expectation should be as large as possible and this is what (C) implies.

In general, there is no unified way to construct an e-variable, but in the situation we are considering, the following theorem gives a representation of any valid e-variable:

\begin{thm}\label{e-value}
For any $\pi\in(0,1)$ and $i\in\cH_0(\pi)$, consider 
\begin{align}
e_{t-1}(x)=\lambda_{it}x+1-\pi\lambda_{it},~~~x\in\{0,1\}, \label{e(X)}
\end{align}
where $\{\lambda_{it}\}$ is an $\cF_{t-1}$-adapted random sequence such that $\lambda_{it}\in[0,1/\pi)$ a.s.\ for all $t\in[T]$. Then the following statements hold:
\begin{enumerate}
\item[(i)] $e_{t-1}(X_{it})$ is a valid $e$-variable for any $t\in[T]$; 
\item[(ii)] Any valid $e$-variable can be represented as $e_{t-1}(X_{it})$.
\end{enumerate}
Moreover, for any $\pi\in(0,1)$ and $i\in\cH_0(\pi)$, $E_{it}=\prod_{s=1}^{t}e_{s-1}(X_{is})$ for each $t\in[T]$ and $E_i = E_{i,T\land\tau_{i}}$ given in Procedure \ref{proc:e} are $e$-variables. 
\end{thm}

The form of $e_{t-1}(X_{it})$ in \eqref{e(X)} is a version of valid $e$-variables. 
For example, $\lambda_{it}X_{it}^2+1-\pi\lambda_{it}$ and $(\lambda_{it}+1-\pi\lambda_{it})^{X_{it}}(1-\pi\lambda_{it})^{1-X_{it}}$ are equivalent to \eqref{e(X)} with $x=X_{it}$.

An immediate consequence of Theorem \ref{e-value} is that $P_i^{\normalfont{\textsf{e}}}=1/E_i$ is indeed a $p$-variable. 
Although Theorem \ref{e-value} justifies the construction of $e$-variables for each $i\in\cH_0(\pi)$, we also have to explain the reason why we multiply them in the proposed method. We say that $e$-variables $Y_1,\dots,Y_T$ are \textit{sequential} if $\E[Y_t\mid Y_1,\dots,Y_{t-1}]\leq1$ a.s.\ for all $t\in\{1,\dots,T\}$. By Condition (B) in Definition \ref{def:valid}, we can check that $E_{i1},\dots,E_{iT}$ are sequential for each $i\in[n]$. Additionally, we say that the functions of sequential $e$-variables to produce an $e$-variable are \textit{se-merging functions}. Following Section 4 in \cite{vovk2021values}, the product function of sequential $e$-variables weakly dominates every other se-merging function. This gives a theoretical justification to use $E_{it}=\prod_{s=1}^{t}e_{s-1}(X_{is})$ in our method.



In the next theorem, we state that $P_i^{\normalfont{\textsf{m}}}$ defined in \eqref{p-value} and $P_i^{\normalfont{\textsf{m}}}h_n$ are also $p$-variables under our setting. Recall $h_n=\sum_{k=1}^n1/k$. 
\begin{thm}\label{misspecifiedp-value}
For any $\pi\in(0,1)$ and $i\in\cH_0(\pi)$, both $P_i^{\normalfont{\textsf{m}}}$ and $P_i^{\normalfont{\textsf{m}}}h_n$ 
are p-variables. 
\end{thm}

\subsection{Theory for the FDR control}\label{fdrcontrol}
Once we obtain an $e$-variable for each hypothesis, \cite{wang2022false} show that the e-BH procedure, which applies the BH for $\{P_i^{\normalfont{\textsf{e}}}\}$ with $P_i^{\normalfont{\textsf{e}}}=1/E_i$, controls the FDR less than or equal to the pre-determined level, regardless of the dependence structure among them. Since the BH with our $\{P_i^{\normalfont{\textsf{e}}}\}$ arrives at the e-BH, the FDR control is guaranteed for each $n$ and $T$. 
\begin{thm}\label{thm:fdr}
For any level $\alpha\in[0,1]$, the BH procedure with $\{P_i^{\normalfont{\textsf{e}}}\}$ controls the $\text{FDR}(\pi)$ at most $\alpha|\cH_0(\pi)|/n$.
\end{thm}

Theorem \ref{thm:fdr} holds for any $\lambda_{it}$ in $P_i^{\normalfont{\textsf{e}}}$. 
On the other hand, the BH with $P_i^{\normalfont{\textsf{m}}}$'s cannot necessarily control the FDR owing to the unknown dependence structure among them. However, the BH with $\{P_i^{\normalfont{\textsf{m}}}h_n\}$, which is equivalent to the BY with $\{P_i^{\normalfont{\textsf{m}}}\}$, is conservative enough to control the FDR under any dependence. 

\begin{thm}\label{thm:fdr2}
For any level $\alpha\in[0,1]$, the BH procedure with $\{P_i^{\normalfont{\textsf{m}}}h_n\}$ controls the $\text{FDR}(\pi)$ at most $\alpha|\cH_0(\pi)|/n$.
\end{thm}

\subsection{Theory for the power guarantee}\label{powerguarantee}

We next consider when the power tends to unity. In a general multiple testing framework, this power guarantee is difficult because it requires simultaneous rejection of nulls for all $i\in\cH_1(\pi)$. However, we can stop updating each $e$-variable when it becomes large enough to reject $H_{0i}$. This means that the condition for the power guarantee reduces to \textit{hypothesis-wise} argument. 

\begin{con}\label{powercondition}
    For any $i\in\cH_1(\pi)$, we have $\Pro(\tau_{i}(\pi)<\infty)\to 1$.
\end{con}

\begin{thm}\label{power}
Suppose that Condition \ref{powercondition} holds. Then for any $\pi\in(0,1)$, the BH with $P_i^{\normalfont{\textsf{e}}}$'s achieves $\Pwr(\pi)\to 1$.
\end{thm}

Whether Condition \ref{powercondition} holds or not depends on the construction of $\lambda_{it}$. Recall that $\lambda_{it}$ can be any random variable for the FDR control as long as it is  $\cF_{t-1}$-measurable and satisfies $0\leq\lambda_{it}<1/\pi$ for each $i\in\cH_0(\pi)$. Therefore, we can construct ``tailored'' $\lambda_{it}$ utilizing past information $\cF_{t-1}$ and domain knowledge for power gain. Roughly speaking, we should make $\lambda_{it}$ large when we believe $X_{it}=1$ occurs, leading to a large $e$-variables and thus a high power.  

Here, we give a set of sufficient conditions for Condition \ref{powercondition} and an example of $\lambda_{it}$ when we believe that for each $i$ the time series $\{X_{it}\}$ admits the law of large numbers. 


\begin{con}\label{cond:3}
Both $T$ and $n$ diverge with $(\log n)/T\to0$. 
\end{con}

\begin{con}\label{cond:2}
For each $i\in\cH_1(\pi)$ with given $\pi\in(0,1)$, there exists some constant $\pi_i\in(\pi,1]$ such that $\bar{X}_{iT}:=\sum_{s=1}^{T}X_{is}/T\xrightarrow{p}\pi_i$ as $T\to\infty$.
\end{con}



\begin{con}\label{cond:4}
For each $i\in\cH_1(\pi)$ with given $\pi\in(0,1)$, there exists some constant $\lambda_i\in(0,\lambda_i^*]\cap(0,1/\pi)$ with $\lambda_i^*:=(\pi_i-\pi)/\{\pi(1-\pi)\}$ such that  $\lambda_{iT}\xrightarrow{p}\lambda_i$. 
\end{con}

Condition \ref{cond:2} holds if $\{X_{it}\}$ is stationary and the serial dependence is sufficiently weak. More specifically, for any $i\in\cH_1(\pi)$, a sufficient condition for $\bar{X}_{iT}\xrightarrow{p}\pi_i$ as $T\to\infty$ is that $\{X_{it}\}$ is weakly stationary with $\E[X_{it}]=\pi_i$ and $\sum_{h=0}^\infty |\Cov{(X_{it},X_{i,t-h})}|<\infty$ holds. 
We explain Condition \ref{cond:4}. 
As observed in the proof, Condition \ref{cond:2} and $\lambda_{iT}\xrightarrow{p}\lambda_i$ for some $\lambda_i$ imply 
\begin{align}\label{logE}
\frac{1}{T}\log E_{iT} = \log (\lambda_i+1-\pi\lambda_i)^{\pi_i}(1-\pi\lambda_i)^{1-\pi_i} + o_p(1)
\end{align}
uniformly in $\lambda_i$. This function is maximized at $\lambda_i^*$. Roughly speaking, Condition \ref{cond:4} says that $\lambda_{it}$ should be designed to be asymptotically at most $\lambda_i^*$. 

\begin{thm}\label{lln}
Conditions \ref{cond:3}--\ref{cond:4} imply Condition \ref{powercondition}.    
\end{thm}

We should design $\lambda_{it}$ that converges to $\lambda_i^*$ in probability because larger $E_{iT}$ leads to higher power. An explicit construction of $\lambda_{it}$ is given as follows. Let $\bar{X}_{i0}=0$. 

\begin{prop}\label{prop:lambda}
For each $i\in\cH_1(\pi)$ with given $\pi\in(0,1)$, define 
\begin{align}\label{eq:lambda}
\lambda_{it}=\min\left\{0\lor\frac{\bar{X}_{i,t-1}-\pi}{\pi(1-\pi)},~\bar{\lambda}\right\}~\text{ for any }~\bar{\lambda}\in(0,1/\pi).
\end{align}
Then, Condition \ref{cond:2} implies Condition \ref{cond:4} with $\lambda_i=\min\{\lambda_i^*,\bar{\lambda}\}$.
\end{prop}

In consequence, by Theorem \ref{lln}, Conditions \ref{cond:3} and \ref{cond:2}  together with \eqref{eq:lambda} imply Condition \ref{powercondition}. This means that the power tends to unity asymptotically, regardless of the choice of  $\bar{\lambda}\in(0,1/\pi)$. Under a finite sample situation, however, we should use $\bar{\lambda}$ that is close to $1/\pi$, like $\bar{\lambda}=1/\pi-0.01$. If $\bar{\lambda}$ is large enough to satisfy $\lambda_i^*\leq\bar{\lambda}$, we have $\lambda_{iT}\xrightarrow{p}\lambda_i=\lambda_i^*$, which will lead to the highest power as mentioned above.


We have considered the power of e-BH with our $\{P_i^{\normalfont{\textsf{e}}}\}$ so far. Meanwhile, the BH with either $\{P_i^{\normalfont{\textsf{m}}}\}$ or $\{P_i^{\normalfont{\textsf{m}}}h_n\}$ is also asymptotically powerful under the same conditions.

\begin{thm}\label{pwr2}
For any $\pi\in(0,1)$, suppose that Conditions \ref{cond:3} and \ref{cond:2} are satisfied. 
Then the BH with either $\{P_i^{\normalfont{\textsf{m}}}\}$ or $\{P_i^{\normalfont{\textsf{m}}}h_n\}$ achieves $\Pwr(\pi)\to 1$.
\end{thm}

By Theorems \ref{power} and \ref{pwr2}, the BH with our $p$-variables can achieve asymptotic power one. However, the power in a finite sample situation might not be high. For a power enhancement of the BY and e-BH in such cases, the following randomization is effective.

\begin{rem}\normalfont
Let $P_i$ denote either $P_i^e$ or $P_i^m h_n$. We can make the BH with $\{P_i\}$ more powerful via randomization. Generating $U\sim U(0,1)$ that is independent of $\{P_i\}$, we use the threshold defined as
\begin{align}
\tilde{K}=\max\left\{i\in[n]:P_i\leq\frac{\alpha(\lfloor i/U\rfloor)}{n}\right\},
\end{align}
instead of $K$ in BH. Then the FDR is still controlled under any dependence structure; see \cite{xu2023more} for more detail. Since the threshold becomes at least as large as the original one, the power will gain. 
\end{rem}

\section{Simulations}\label{sec:Simulations}

We consider four cases: i.i.d.\ Bernoulli in Section \ref{iidbernoulli}, logistic regression model in Section \ref{logistic}, level shift model in Section \ref{level}, and periodic data in Section \ref{sec:periodic}. For each case, we fix the number of edges $n=300$, the parameter of the hypotheses $\pi=0.1$, and the FDR level $\alpha=0.1$ while vary the number of $\pi$-connectable edges $|\cH_1(\pi)|\in\{30,60,90,120,150\}$ and the number of networks $T\in\{100,200,300,400,500\}$. We compare the performances of $p$-variables $P_i^m$, $P_i^mh_n$, and $P_i^e$ based on $500$ replications. Here, we use \eqref{eq:lambda} with $\bar{\lambda}=1/\pi-0.01$ to construct $P_i^e$.

\subsection{i.i.d.\ Bernoulli}\label{iidbernoulli}
For $i\in[n]$ and $t\in[T]$, we generate $X_{it}\sim \text{i.i.d.}Ber(\pi_i)$. Since the Bernoulli parameters are non-random in this case, our multiple testing problem reduces to
\begin{align}
    H_{0i}:\pi_i\leq \pi\text{ versus }H_{1i}:\pi_i>\pi~\text{for each }i\in[n].
\end{align}
We create the $p$-variables $P_i^e$ and $P_i^m$ as in Procedure \ref{proc:e} and \eqref{p-value}, respectively. Since the resulting $p$-variables are independent under the current setting, the BH procedure controls the FDR as the BY and e-BH do. In our simulation, we set $\pi_i=0.1$ for $i\in\cH_0(\pi)$ and $\pi_i=0.15$ for $i\in\cH_0(\pi)$. 

Figures \ref{fig:FDR1} and \ref{fig:PWR1} show the FDR and power, respectively. Since there is no dependence among edges, each procedure controls the FDR. The BH procedure is the most powerful as expected. Moreover, the power of the e-BH procedure is less than that of the BY procedure, which may be because the quality of $p$-variables are better than $e$-variables.

\begin{figure}[H]
  \begin{minipage}[b]{0.32\linewidth}
    \centering
    \includegraphics[keepaspectratio,scale=0.26,bb=700 0 0 800]{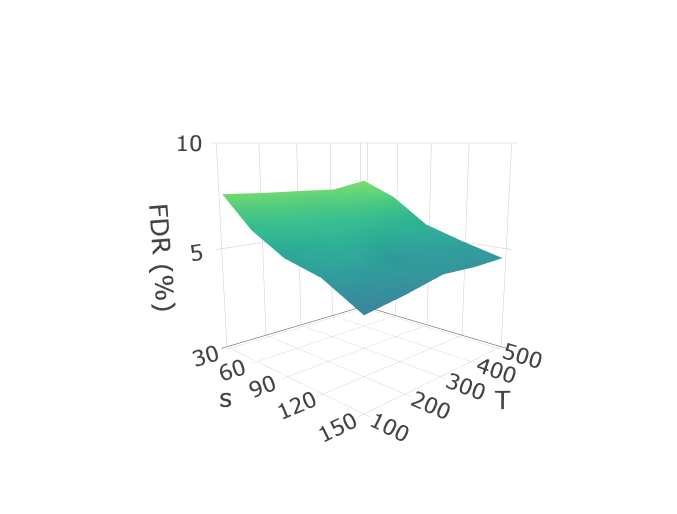}
  \end{minipage}
  \begin{minipage}[b]{0.32\linewidth}
    \centering
    \includegraphics[keepaspectratio,scale=0.26,bb=700 0 0 800]{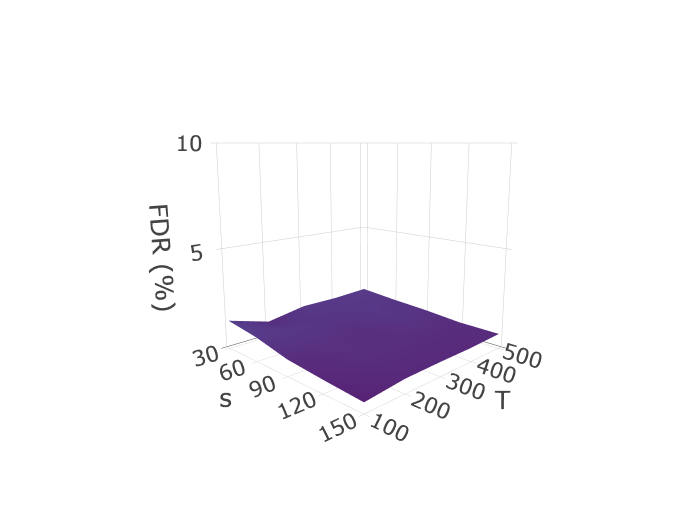}
  \end{minipage}
  \begin{minipage}[b]{0.32\linewidth}
    \centering
    \includegraphics[keepaspectratio,scale=0.26,bb=700 0 0 800]{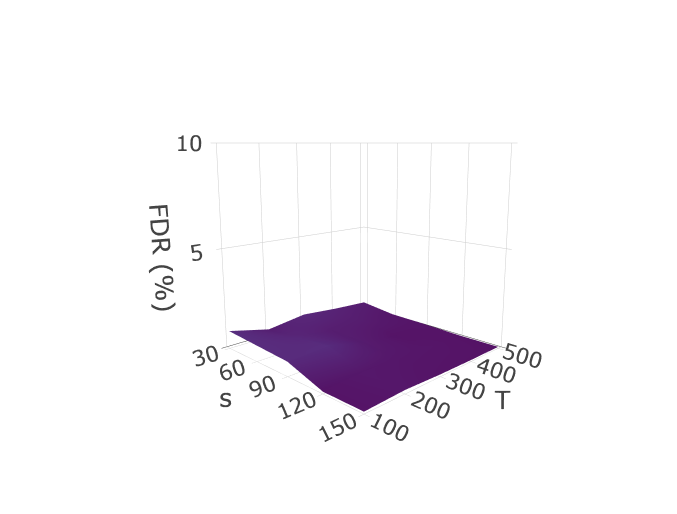}
  \end{minipage}
  \caption{FDR for $P_i^m$ (left), $P_i^mh_n$ (middle), and $P_i^e$ (right) for i.i.d. data.}
  \label{fig:FDR1}
\end{figure}

\begin{figure}[H]
  \begin{minipage}[b]{0.32\linewidth}
    \centering
\includegraphics[keepaspectratio,scale=0.26,bb=700 0 0 800]{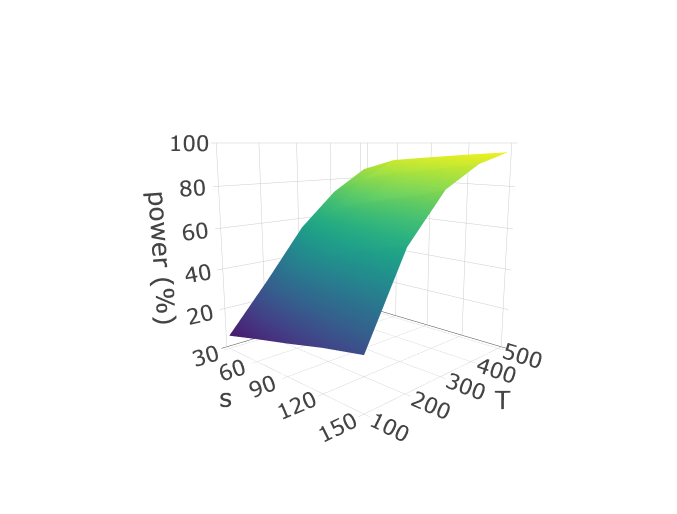}
  \end{minipage}
  \begin{minipage}[b]{0.32\linewidth}
    \centering
    \includegraphics[keepaspectratio,scale=0.26,bb=700 0 0 800]{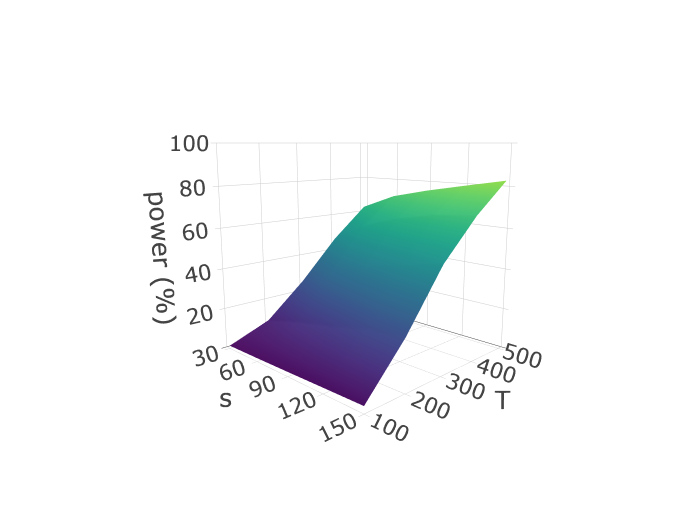}
  \end{minipage}
  \begin{minipage}[b]{0.32\linewidth}
    \centering
    \includegraphics[keepaspectratio,scale=0.26,bb=700 0 0 800]{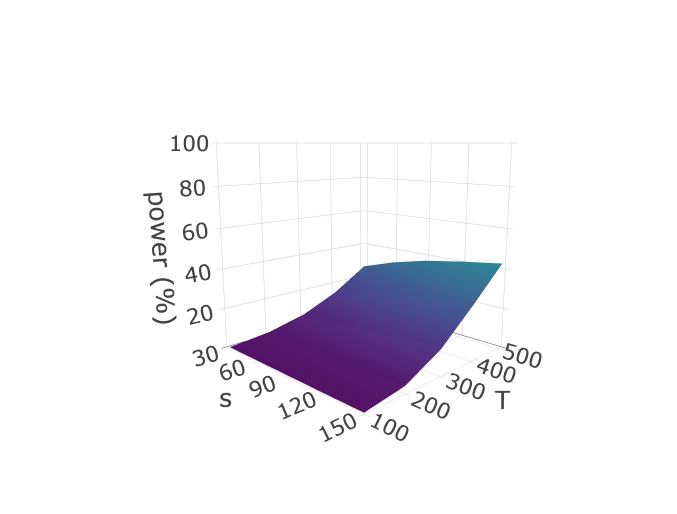}
  \end{minipage}
  \caption{Power for $P_i^m$ (left), $P_i^mh_n$ (middle), and $P_i^e$ (right) for i.i.d. data.}
  \label{fig:PWR1}
\end{figure}

\subsection{Logistic regression model}\label{logistic}

We next consider the following logistic model that has a complex dependence structure. For each $i\in[n]$ and $t\in[T]$, let
\begin{align}
    \xi_{it}^{\beta}:=\E[X_{it}\mid X_{t-1}]=\frac{1}{1+\exp\left\{-(\beta_{i0}+\beta_{i}^\top X_{t-1})\right\}},
\end{align}
where $(\beta_{i0},\beta_{i}^\top)^\top\in\bbR^{1+n}$ is a coefficient vector and $X_t=(X_{1t},\dots,X_{nt})^\top$. We constrain all the coefficients to be negative, and therefore $\xi_{it}^{\beta}$ takes the maximum $1/(1+\exp\{-\beta_{i0}\})$ when $X_{t-1}=(0,\dots,0)^\top$. Then our multiple testing problem reduces to $H_{0i}:1/(1+\exp\{-\beta_{i0}\})\leq\pi~\text{versus}~H_{1i}:1/(1+\exp\{-\beta_{i0}\})>\pi$, or equivalently,
\begin{align}
    H_{0i}:\beta_{i0}\leq\log\frac{\pi}{1-\pi}~~~\text{versus}~~~H_{1i}:\beta_{i0}>\log\frac{\pi}{1-\pi}.
\end{align}
In our simulation, for $i\in\cH_0(\pi)$, we set any coefficient to be $\log(\pi/(1-\pi))$ and for $i\in\cH_1(\pi)$, we set any coefficient to be $(2/3)\log(\pi/(1-\pi))$.  

The complex dependence structure of the model in both $i$ and $t$ makes the BH with $P_i^m$ have no FDR guarantee. However, as in Figure \ref{fig:FDR2}, it seems that the FDR is controlled under a quite low level. This is probably because the null distributions are conservative with high probability. On the other hand, Figure \ref{fig:PWR2} shows that $P_i^m$ is the most powerful.

\begin{figure}[H]
  \begin{minipage}[b]{0.32\linewidth}
    \centering
    \includegraphics[keepaspectratio,scale=0.26,bb=700 0 0 800]{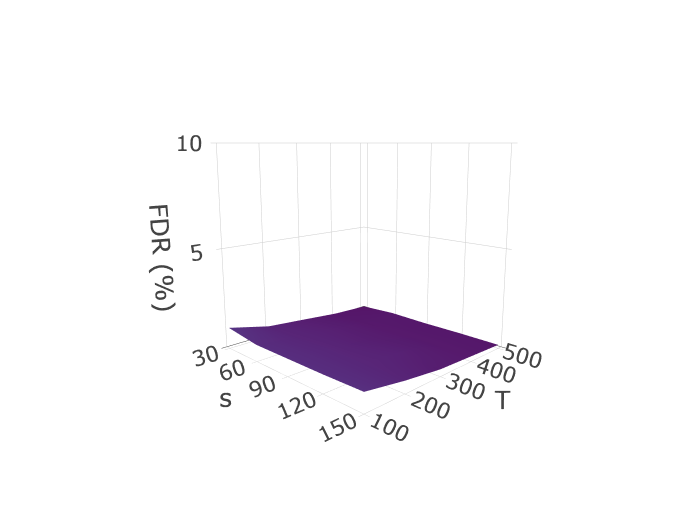}
  \end{minipage}
  \begin{minipage}[b]{0.32\linewidth}
    \centering
    \includegraphics[keepaspectratio,scale=0.26,bb=700 0 0 800]{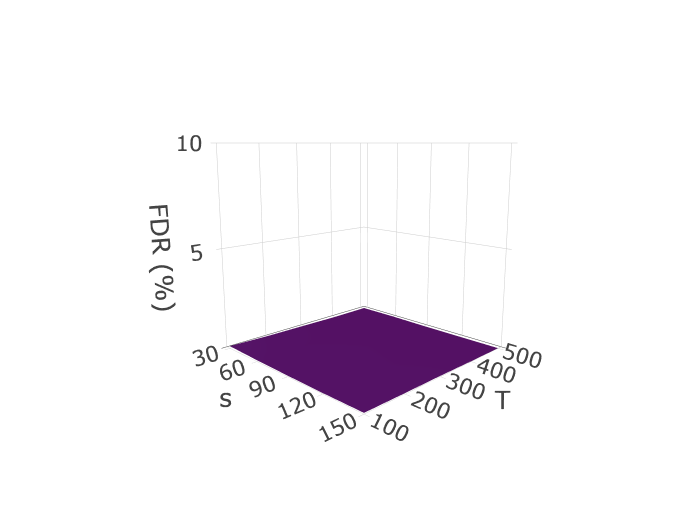}
  \end{minipage}
  \begin{minipage}[b]{0.32\linewidth}
    \centering
    \includegraphics[keepaspectratio,scale=0.26,bb=700 0 0 800]{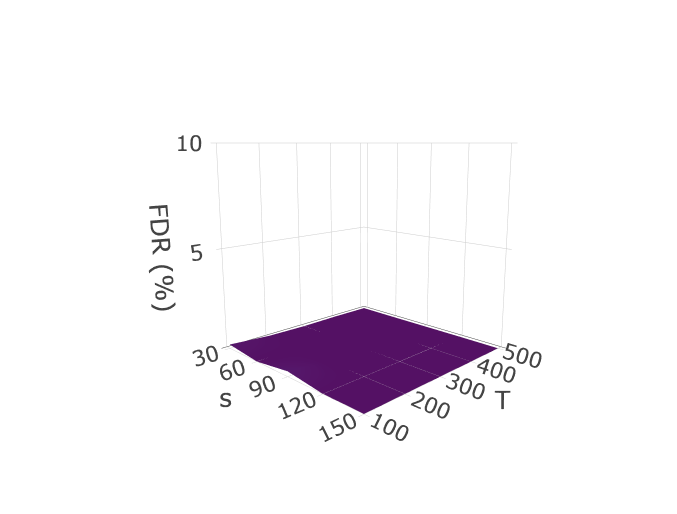}
  \end{minipage}
  \caption{FDR for $P_i^m$ (left), $P_i^mh_n$ (middle), and $P_i^e$ (right) for logistic model.}
  \label{fig:FDR2}
\end{figure}

\begin{figure}[H]
  \begin{minipage}[b]{0.32\linewidth}
    \centering
    \includegraphics[keepaspectratio,scale=0.26,bb=700 0 0 800]{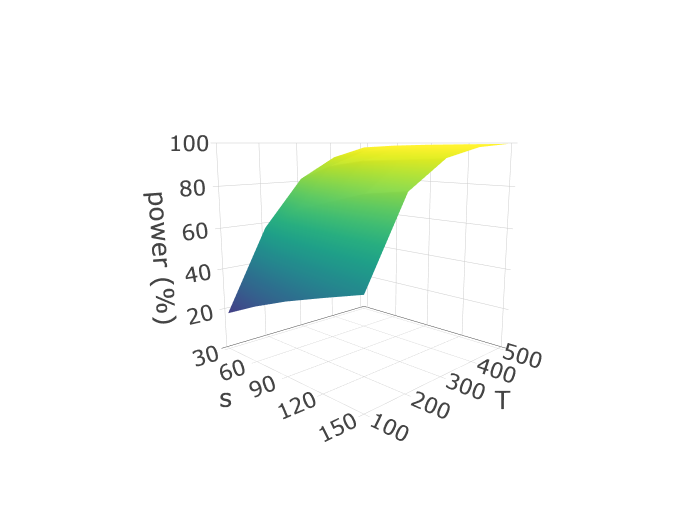}
  \end{minipage}
  \begin{minipage}[b]{0.32\linewidth}
    \centering
    \includegraphics[keepaspectratio,scale=0.26,bb=700 0 0 800]{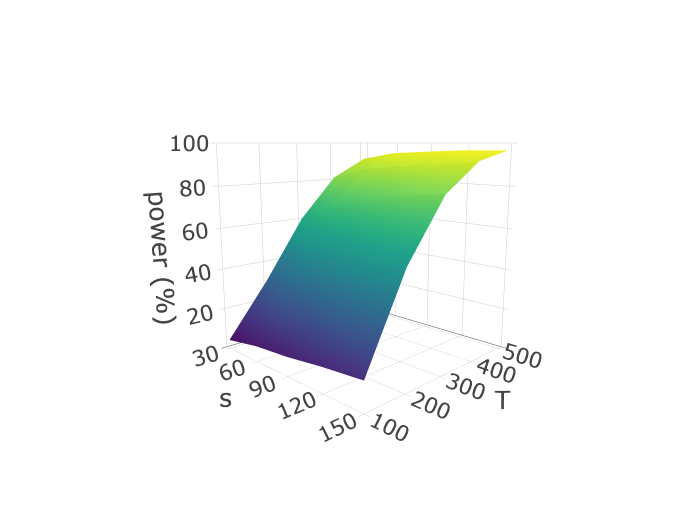}
  \end{minipage}
  \begin{minipage}[b]{0.32\linewidth}
    \centering
    \includegraphics[keepaspectratio,scale=0.26,bb=700 0 0 800]{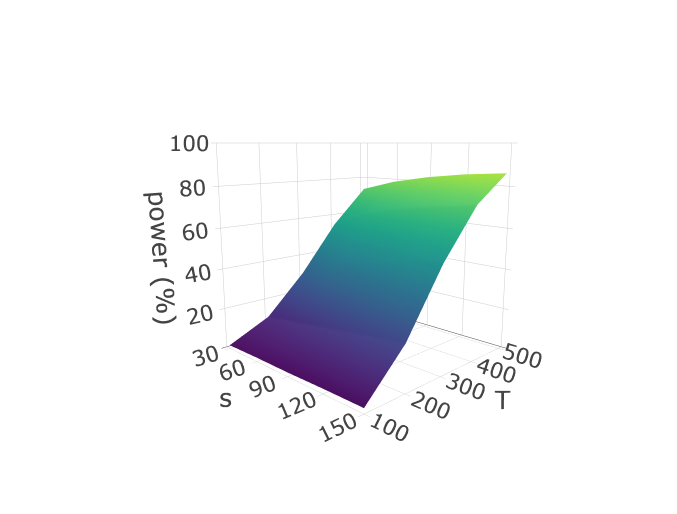}
  \end{minipage}
  \caption{Power for $P_i^m$ (left), $P_i^mh_n$ (middle), and $P_i^e$ (right) for logistic model.}
  \label{fig:PWR2}
\end{figure}

\subsection{Level shift model}\label{level}
One of the advantages of $P_i^e$ is that it can stop as soon as sufficient evidence against the null accumulates. To demonstrate this benefit, we consider the following lebel shift model. For each $i\in[n]$, we suppose
\begin{align}
    X_{it}\sim \text{indep.}
    \begin{cases}
        \text{Ber}(\pi_{ia}) &\text{for }~t\leq T_0, \\
        \text{Ber}(\pi_{ib}) &\text{for }~t>T_0,
    \end{cases}
\end{align}
where $\pi_{ia}$ and $\pi_{ib}$ are some constants in $(0,1)$. Then our multiple testing problem reduces to\begin{align}
    H_{0i}:\max\{\pi_{ia},\pi_{ib}\}\leq\pi~~~\text{versus}~~~H_{1i}:\max\{\pi_{ia},\pi_{ib}\}>\pi.
\end{align}
In our simulation, we set $T_0=30$. For $i\in\cH_0(\pi)$ and $i\in\cH_1(\pi)$, we set $(\pi_{ia},\pi_{ib})=(\pi,\pi/2)$ and $(\pi_{ia},\pi_{ib})=(0.5,\pi/2)$, respectively. This means that the strong signal will be observed in the first period.  As we expected, Figure \ref{fig:PWR3} shows that the BH procedure with $P_i^e$ is the most powerful. On the other hand, $P_i^m$ and $P_i^mh_n$ lead to the low power except for the case of $T=100$. This is because these $p$-values are based on the sample mean with all the observations, and not allowed to be stopped.

\begin{figure}[H]
  \begin{minipage}[b]{0.32\linewidth}
    \centering
    \includegraphics[keepaspectratio,scale=0.26,bb=700 0 0 800]{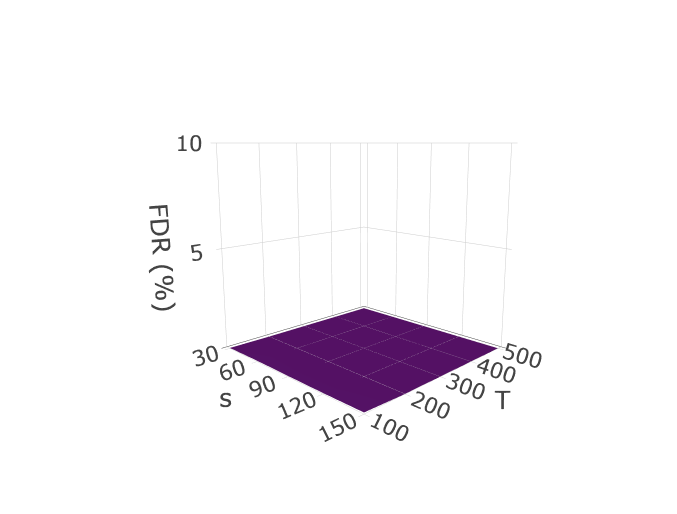}
  \end{minipage}
  \begin{minipage}[b]{0.32\linewidth}
    \centering
    \includegraphics[keepaspectratio,scale=0.26,bb=700 0 0 800]{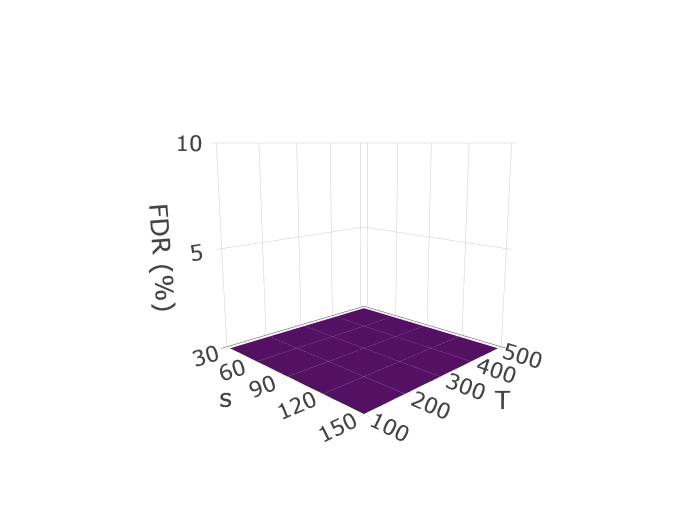}
  \end{minipage}
  \begin{minipage}[b]{0.32\linewidth}
    \centering
    \includegraphics[keepaspectratio,scale=0.26,bb=700 0 0 800]{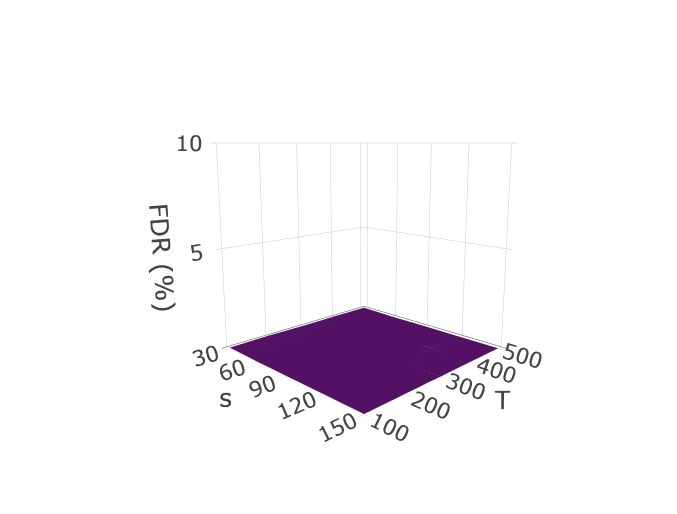}
  \end{minipage}
  \caption{FDR for $P_i^m$ (left), $P_i^mh_n$ (middle), and $P_i^e$ (right) for level shift model.}
  \label{fig:FDR3}
\end{figure}

\begin{figure}[H]
  \begin{minipage}[b]{0.32\linewidth}
    \centering
    \includegraphics[keepaspectratio,scale=0.26,bb=700 0 0 800]{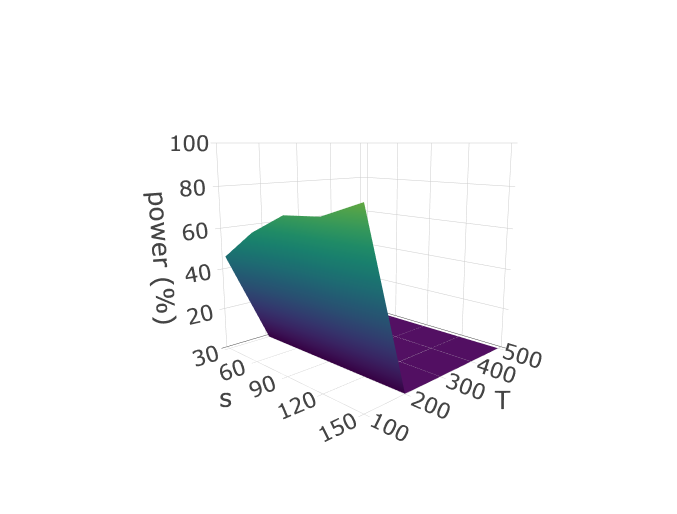}
  \end{minipage}
  \begin{minipage}[b]{0.32\linewidth}
    \centering
    \includegraphics[keepaspectratio,scale=0.26,bb=700 0 0 800]{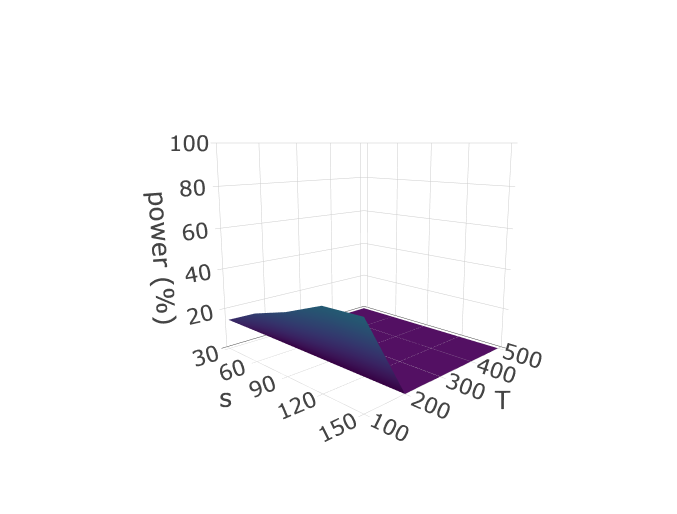}
  \end{minipage}
  \begin{minipage}[b]{0.32\linewidth}
    \centering
    \includegraphics[keepaspectratio,scale=0.26,bb=700 0 0 800]{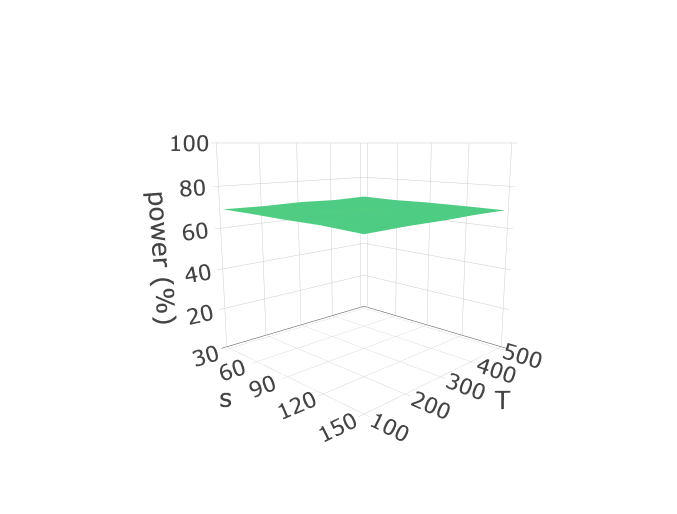}
  \end{minipage}
  \caption{Power for $P_i^m$ (left), $P_i^mh_n$ (middle), and $P_i^e$ (right) for level shift model.}
  \label{fig:PWR3}
\end{figure}

\subsection{Periodic data}\label{sec:periodic}
One of the advantages of our $e$-variable is that it can incorporate domain knowledge for power improvement. As an example, we assume that the edges in $\cH_1(\pi)$ are more likely to be connected periodically, and demonstrate the power gain if we know it. Periodic data often appear in many applications. For instance, the traffic volume is usually higher on weekends.

For each $i\in[n]$, we generate $X_{it}\sim\textit{Ber}(\pi_{it})$ independently, where the parameters are further sampled from the following: 
For $i\in\cH_0(\pi)$, let
\begin{align}
\pi_{it} \sim 
\begin{cases}
    U(0,0.05)   & \text{ for } t\in[T] \setminus \{5,10,15,\dots,T\};\\
    U(0.05,0.1) &\text{ for } t\in \{5,10,15,\dots,T\},
\end{cases}
\end{align}
whereas for $i\in\cH_1(\pi)$, let
\begin{align}
\pi_{it} \sim 
\begin{cases}
    U(0,0.05)   & \text{ for } t\in[T] \setminus \{5,10,15,\dots,T\};\\
    U(0.2,0.3) &\text{ for } t\in \{5,10,15,\dots,T\}.
\end{cases}
\end{align}

According to this distributional structure as the domain knowledge, we may construct the hyper parameter in $P_i^e$ as follows:
\begin{align}\label{knowledge}
    \lambda_{it}=
    \begin{cases}
0.1 & \text{ for } t\in[T] \setminus \{5,10,15,\dots,T\};\\
1.5 &\text{ for } t\in \{5,10,15,\dots,T\}.
\end{cases}
\end{align}
which can be a better alternative to \eqref{eq:lambda}. Since the resulting $e$-variable is obtained by the product of $\{e_{t-1}(X_{it})\}$, each should be as large as possible. To this end, small (large) $\lambda_{it}$ should be assigned when $X_{it}=0$ ($=1$) is expected. Indeed, as $\lambda_{it}$ gets larger, $e_{t-1}(0)\leq1$ gets smaller and $e_{t-1}(1)\geq1$ gets larger.

Figures \ref{fig:FDR4} and \ref{fig:PWR4} indicate that the BH applying to $\{P_i^m\}$, $\{P_i^mh_n\}$, and $\{P_i^e\}$ with \eqref{eq:lambda} find no edges, resulting in zero \% of the FDR and power for every $(s,T)$. On the contrary, Figure \ref{fig:periodebhknow} show that the e-BH along with $\lambda_{it}$ defined in \eqref{knowledge} has a high power especially when $T$ is large while the FDR is under control. This observation suggests that our $e$-variable has the potential to increase power by incorporating the domain knowledge.
\begin{figure}[H]
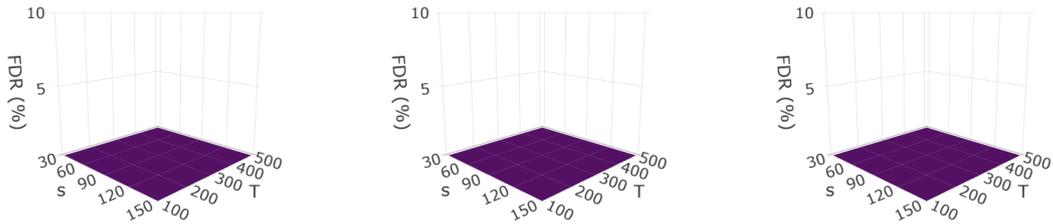

  \begin{minipage}[b]{0.32\linewidth}
    \centering
    \includegraphics[keepaspectratio,scale=0.26,bb=700 0 0 800]{figure/period_bh_fdr_a.png}
  \end{minipage}
  \begin{minipage}[b]{0.32\linewidth}
    \centering
    \includegraphics[keepaspectratio,scale=0.26,bb=700 0 0 800]{figure/period_by_fdr_a.png}
  \end{minipage}
  \begin{minipage}[b]{0.32\linewidth}
    \centering
    \includegraphics[keepaspectratio,scale=0.26,bb=700 0 0 800]{figure/period_ebhnot_fdr_a.png}
  \end{minipage}
  \caption{FDR for $P_i^m$ (left), $P_i^mh_n$ (middle), and $P_i^e$ (right) without knowledge for periodic data.}
  \label{fig:FDR4}
\end{figure}

\begin{figure}[H]
  \begin{minipage}[b]{0.32\linewidth}
    \centering
    \includegraphics[keepaspectratio,scale=0.26,bb=700 0 0 800]{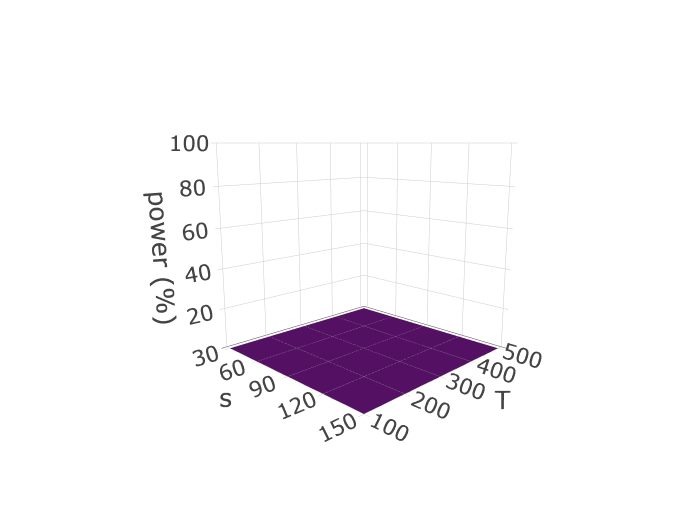}
  \end{minipage}
  \begin{minipage}[b]{0.32\linewidth}
    \centering
    \includegraphics[keepaspectratio,scale=0.26,bb=700 0 0 800]{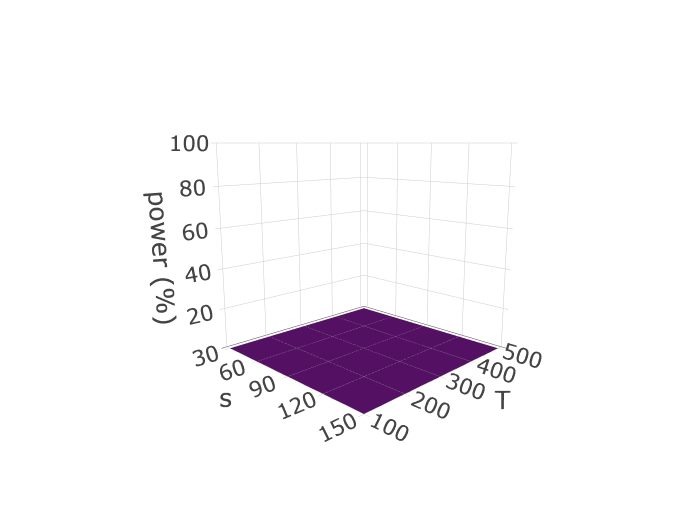}
  \end{minipage}
  \begin{minipage}[b]{0.32\linewidth}
    \centering
    \includegraphics[keepaspectratio,scale=0.26,bb=700 0 0 800]{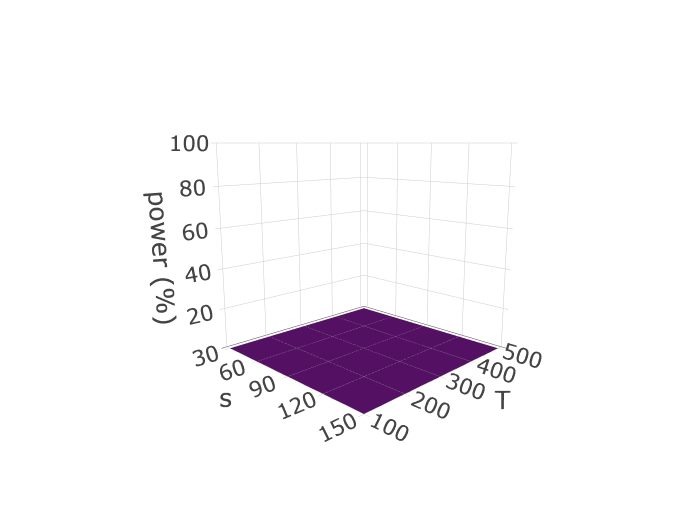}
  \end{minipage}
  \caption{Power for $P_i^m$ (left), $P_i^mh_n$ (middle), and $P_i^e$ (right) without knowledge for periodic data.}
  \label{fig:PWR4}
\end{figure}

\begin{figure}[H]
  \begin{minipage}[b]{0.5\linewidth}
    \centering
    \includegraphics[keepaspectratio,scale=0.4,bb=700 0 0 800]{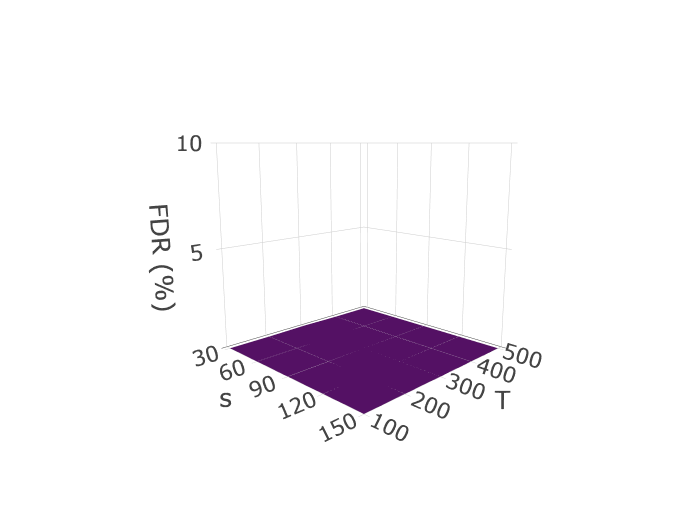}
  \end{minipage}
  \begin{minipage}[b]{0.5\linewidth}
    \centering
    \includegraphics[keepaspectratio,scale=0.4,bb=700 0 0 800]{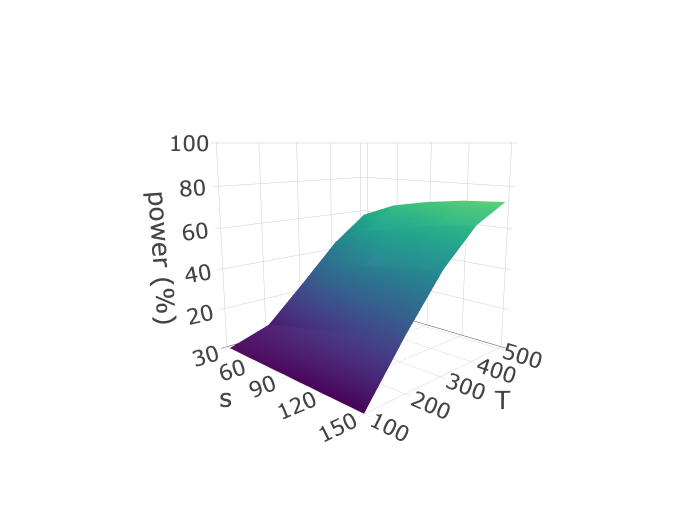}
  \end{minipage}
  \caption{FDR (left) and power (right) for $P_i^e$ using knowledge of periodic data. }
  \label{fig:periodebhknow}
\end{figure}

\section{Application to SNS data}\label{sec:Real data}
Suppose that we have time series of SNS data, which record who sent and received the massages, observed per day. Thus the SNS data can be represented by a time series network which contains $n$ (the number of the pairs of a sender and receiver) potential edges  for each day $t\in[T]$. For each $(i,t)\in[n]\times[T]$, the $i$th edge is connected at day $t$ if any massage is sent between a certain pair of the sender and receiver. In this case, we denote $X_{it}=1$ and otherwise $X_{it}=0$. 

We use the data from \cite{panzarasa2009patterns}, where $n=1899^2$ and $T=195$. To study the dynamics of SNS, we are interested in detecting the $\pi$-connected network. We apply the BH (with $P_i^m$), BY (i.e., BH with $P_i^mh_n$), and e-BH (i.e., BH with $P_i^e$) for $\pi\in\{0.01,0.02\}$ and $\alpha=0.1$. As for $P_i^e$, we use $\lambda_{it}$ given in \eqref{eq:lambda} with $\bar{\lambda}=1/\pi-0.01$. 

For many human communication dynamics, it is known that a distribution of a time interval between two consecutive events tend to have a heavy-tail \citep{vazquez2007impact,holme2012temporal}. Such an observation implies that a very long time is often needed before the same edge connects again. Therefore, even for $i\in\cH_1(\pi)$, the sample proportion of the occurrence of the edge $i$ can be much smaller than $\pi$. In this case, the $p$-value for $H_{0i}$ will be large, and result in low power. On the other hand, thanks to the optional stopping property, $e$-value-based procedures will be a promising option since it can reject $H_{0i}$ as soon as the evidence against it accumulates.

The results are shown in Figures \ref{fig:bh}, \ref{fig:by} and \ref{fig:ebh}, We observe that the detected network gets sparser as $\pi$ becomes larger. In contrast to BY and e-BH, the FDR of BH may not be controlled because there can be a complex dependence structure among edges. Remarkably, the e-BH not only controls the FDR, it also detects the most edges as we expected. This could be explained by the optional stopping property. In fact, the time periods for each edge to be connected are clustered overall.

\begin{figure}[H]
  \begin{minipage}[b]{0.5\linewidth}
    \centering
    \includegraphics[keepaspectratio,scale=0.5,bb=-50 0 300 300]{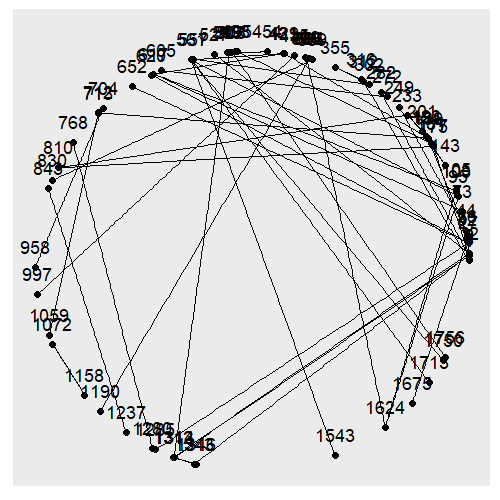}
  \end{minipage}
  \begin{minipage}[b]{0.5\linewidth}
    \centering
    \includegraphics[keepaspectratio,scale=0.5,bb=-50 0 300 300]{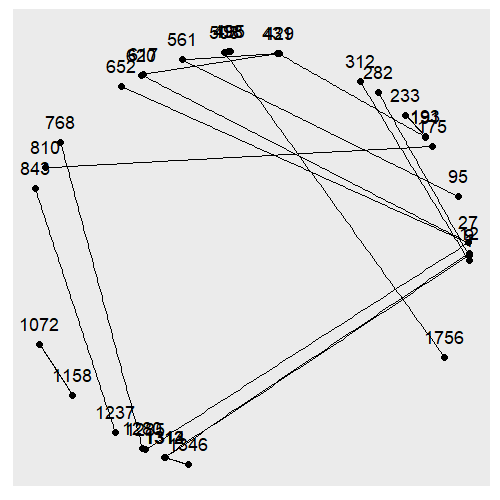}
  \end{minipage}
  \caption{Detected networks by BH with $\pi=0.01$ (left) and $\pi=0.02$ (right).}
  \label{fig:bh}
\end{figure}

\begin{figure}[H]
  \begin{minipage}[b]{0.5\linewidth}
    \centering
    \includegraphics[keepaspectratio,scale=0.5,bb=-50 0 300 300]{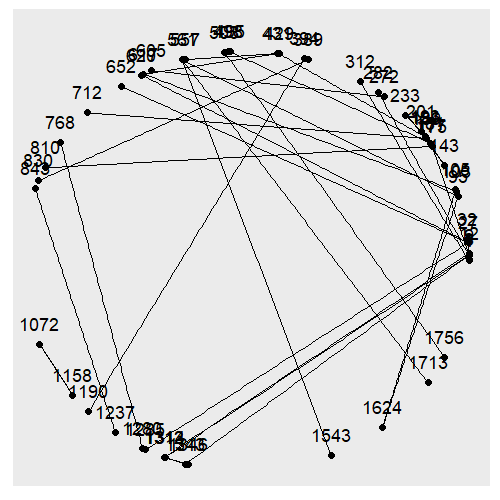}
  \end{minipage}
  \begin{minipage}[b]{0.5\linewidth}
    \centering
    \includegraphics[keepaspectratio,scale=0.5,bb=-50 0 300 300]{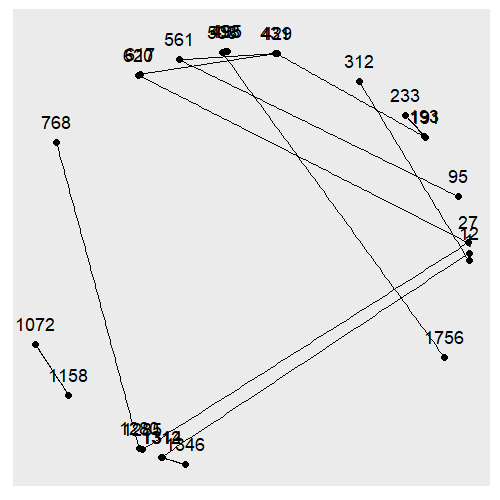}
  \end{minipage}
  \caption{Detected networks by BY with $\pi=0.01$ (left) and $\pi=0.02$ (right).}
  \label{fig:by}
\end{figure}

\begin{figure}[H]
  \begin{minipage}[b]{0.5\linewidth}
    \centering
    \includegraphics[keepaspectratio,scale=0.5,bb=-50 0 300 300]{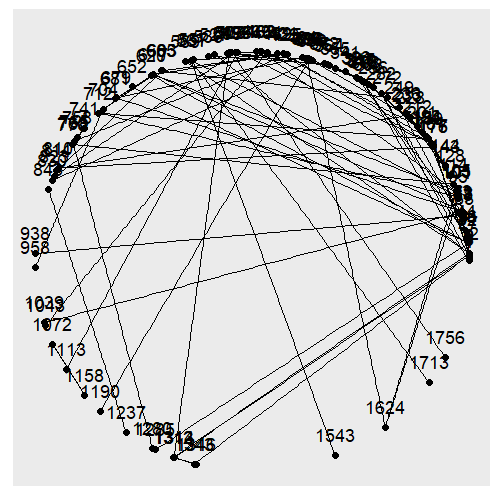}
  \end{minipage}
  \begin{minipage}[b]{0.5\linewidth}
    \centering
    \includegraphics[keepaspectratio,scale=0.5,bb=-50 0 300 300]{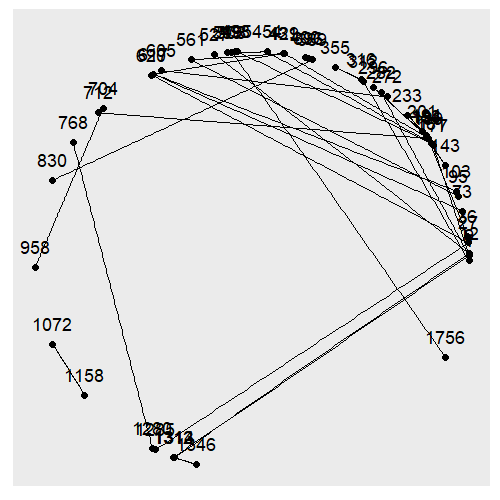}
  \end{minipage}
  \caption{Detected networks by e-BH with $\pi=0.01$ (left) and $\pi=0.02$ (right).}
  \label{fig:ebh}
\end{figure}

\section{Conclusion}\label{sec:Conclusion}

In this paper, we propose a methodology of exploratory data analysis for discovering relationships between variables, which are conceptualized as edges between nodes for visual understanding, in high dimension. While there are various methods for selecting important variables, methodologies for uncovering connectable edges and the networks they form are unprecedented. Networks discovered by our method are worth investigating in more detail. This leads to more efficient analysis to complex and high-dimensional network data, and more meaningful implication for building a new network science. 

We summarize our contributions here. We first define the $\pi$-connectable edges and network, and set them as our object of detection. Then we formalize the detection problem as a multiple testing. To implement it, we construct $p$-variables $P_i^m$, $P_i^mh_n$, and $P_i^e$, which is our main methodological contribution, and apply the BH to them. For $P_i^mh_n$ and $P_i^e$, we show that FDR is controlled under arbitrary dependence structure, thus we can discover edges with reproducibility via FDR control. Moreover, under mild conditions, we also show that the power tends to unity. We demonstrate the performances of each $p$-variable by simulations and real data example. If we have some knowledge of the network, the use of $P_i^e$ is recommended since it can incorporate the knowledge flexibly.

Finally, we indicate a few directions for future studies. (i) We have restricted our attention to network whose nodes are non-random and edges are realized as binary variables. However, recently we can obtain network data with random nodes and/or real-valued edges such as trade networks. The extension of our method for such more complex networks is interesting. (ii) The $\pi$-connectable edge is defined by the $\ell_\infty$-norm of $\xi_{it}$, meaning that they are connectable at least a certain period. However, one may want to consider edges likely to be connected over entire periods. Thus the extension to using another norm, such as $\ell_2$-norm, is meaningful. (iii) We have to mention about the potential to improve the power by incorporating the dependence structure among edges. We often have some knowledge about it, such as a block or hierarchical structure. Although our $p$-variables are constructed marginally, it may be possible to construct new $p$-variables by using multiple binary sequences jointly in such cases.

\newpage
\appendix
\section{Appendix: Proofs of the Theoretical Results}

Throughout the proofs, we omit writing a.s.\ when it is clear from the context.

\subsection{Proof of Theorem \ref{e-value}}

\begin{proof}[Proof of Theorem \ref{e-value}.]
First, we show that $e(X_{it};\cF_{t-1})=\lambda_{it}X_{it}+1-\pi\lambda_{it}$, where $\lambda_{it}$ is any $[0,1/\pi)$-valued $\cF_{t-1}$-measurable random variable for $t=1,\dots,T$,  is a valid sequence of  $e$-variables. 
We check if Condition (A) holds. By the definition, we have 
\begin{align}
e_{t-1}(0) = 1-\pi\lambda_{it}, ~~~
e_{t-1}(1) = 1-\pi\lambda_{it} + \lambda_{it} = e_{t-1}(0)+\lambda_{it}.
\end{align}
Because $\lambda_{it}$ takes the values in  $[0,1/\pi)$,  we have $e_{t-1}(0)>0$ and $e_{t-1}(0) \leq e_{t-1}(1)$, which verifies (A). We check if (B) and (C) are true. Since $0\leq\xi_{it}\leq \pi$ for any $i\in\cH_0(\pi)$, we have
\begin{align}
\E\left[e_{t-1}(X_{it})\mid\cF_{t-1}\right] 
&= \lambda_{it}\xi_{it} + 1- \pi\lambda_{it} 
\leq \lambda_{it}\pi + 1- \pi\lambda_{it} =1, 
\end{align}
giving (B). In this inequality, the equality holds when $\{\xi_{it}=\pi\}$ occurs. Therefore, (C) holds.

Next, we show that any valid $e$-variable can be represented as $e_{t-1}(X_{it})=\lambda_{it}X_{it}+1-\pi\lambda_{it}$ for some $\cF_{t-1}$-measurable $\lambda_{it}\in[0,1/\pi)$. 
Without loss of generality, we suppose that $x\mapsto e_{t-1}(x)$ is linear; that is, put $e_{t-1}(x)=a_t x + b_t $ for some $\cF_{t-1}$-measurable random variables $a_t,b_t$. 
By Condition (A), it must be that $0<b_t\leq a_t+b_t$, which means $a_t\geq0$ and $b_t>0$. Moreover, by (B), we must have
\begin{align}
\E\left[e_{t-1}(X_{it})\mid\cF_{t-1}\right]=\xi_{it}a_t+b_t \leq 1.
\end{align}
Since $a_t$ is nonnegative, this is maximized to be $1$ under the event $\{\xi_{it}=\pi\}$, leading to $\pi a_t+b_t=1$. Furthermore, $b_t>0$ implies that $a_t<1/\pi$. Denoting $a_t$ as $\lambda_{it}$, we have $b_t=1-\pi\lambda_{it}$.

Finally, by \cite{grunwald2020safe}, we claim that $E_{i}=E_{i,T\land\tau_i(\pi)}$ is indeed an $e$-variable for any $\pi\in(0,1)$, $i\in\cH_0(\pi)$, and a stopping time $\tau_i$.  
This completes the proof.
\end{proof}

\subsection{Proof of Theorem \ref{misspecifiedp-value}}

\begin{proof}[Proof of Theorem \ref{misspecifiedp-value}]
Let $X_{it}$ denote an edge  with $i\in\cH_0(\pi)$ and $Y_1,\dots,Y_T\sim \text{i.i.d.}\ Ber(\pi)$. By their definitions, for any $(i,t)\in\cH_0(\pi)\times[T]$ and $c\in(0,1]$, we have 
\begin{align}
\Pro(X_{it}\geq c\mid\cF_{t-1})&=\Pro(X_{it}=1\mid\cF_{t-1}) \\
&\leq\pi=\Pro(Y_t=1)=\Pro(Y_t\geq c)
=\Pro(Y_t\geq c\mid\cF_{t-1}).
\end{align}
This is also true for any $c\leq 0$ and $c>1$ with equality since their probabilities are trivially equal to $1$ and $0$. Thus for any $(i,t)\in\cH_0(\pi)\times[T]$ and $c\in\mathbb{R}$, we obtain 
\begin{align}
\Pro\left(X_{it}\geq c\mid\cF_{t-1}\right)\leq\Pro\left(Y_t\geq c\mid\cF_{t-1}\right).\label{ineq:2}
\end{align}

Let $S_i = \sum_{t=1}^{T}X_{it}$ and $S = \sum_{t=1}^{T}Y_{t}$. 
By the inequality in \eqref{ineq:2} and the law of iterated expectations, we have
    \begin{align}\label{stochasticallylarger}
        \Pro\left(S_i\geq c\right)
        &=\E\left[\Pro\left(X_{iT}\geq c-X_{i1}-\dots-X_{i,T-1}\mid X_{i1},\dots,X_{i,T-1}\right)\right]\\
        &\leq\E\left[\Pro\left(Y_{T}\geq c-X_{i1}-\dots-X_{i,T-1}\mid X_{i1},\dots,X_{i,T-1}\right)\right]\\
        &=\Pro\left(Y_{T}\geq c-X_{i1}-\dots-X_{i,T-1}\right)\\
        &=\E\left[\Pro\left(X_{i,T-1}\geq c-X_{i1}-\dots-X_{i,T-2}-Y_T\mid X_{i1},\dots,X_{i,T-2},Y_T\right)\right]\\
        &\leq\E\left[\Pro\left(Y_{T-1}\geq c-X_{i1}-\dots-X_{i,T-2}-Y_T\mid X_{i1},\dots,X_{i,T-2},Y_T\right)\right]\\
        &=\Pro\left(Y_{T-1}+Y_T\geq c-X_{i1}-\dots-X_{i,T-2}\right).
    \end{align}
Repeating this manipulation results in
\begin{align}
\Pro\left(S_i\geq c\right)
\leq
\Pro\left(S\geq c\right). \label{ineq:1}
\end{align}

Recall $P_{i}^{\normalfont{\textsf{m}}}
= \sum_{t=S_i}^{T}\binom Tt
\pi^t(1-\pi)^{T-t}$, which was defined in \eqref{p-value}, and denote by $F_{1-\alpha}=\min\left\{x:\Pro\left(S\geq x\right)\leq\alpha\right\}$ the upper $\alpha$-quantile of $S\sim \text{Bin}(T,\pi)$. Then we obtain 
\begin{align}
\Pro\left(P_{i}^{\normalfont{\textsf{m}}}\leq\alpha\right)
&=\Pro\left(\sum_{t=S_i}^{T}
\binom Tt
\pi^t(1-\pi)^{T-t}\leq\alpha\right)\\
&=\Pro\left(S_i\geq F_{1-\alpha}\right)
\leq\Pro\left(S\geq F_{1-\alpha}\right)
\leq\alpha.
\end{align}
where the first inequality holds by \eqref{ineq:1}. This completes the proof.  
\end{proof}

\subsection{Proof of Theorem \ref{thm:fdr}}

\begin{proof}[Proof of Theorem \ref{thm:fdr}]
The result immediately follows since $\{P_i^e\}$ is a set of $p$-variables that leads to the procedure of \cite{wang2022false}.
\end{proof}

\subsection{Proof of Theorem \ref{thm:fdr2}}

\begin{proof}[Proof of Theorem \ref{thm:fdr2}]
The result immediately follows since $\{P_i^mh_n\}$ is a set of $p$-variables that leads to the procedure of \cite{benjamini2001control}.  
\end{proof}

\subsection{Proof of Theorem \ref{power}}

\begin{proof}[Proof of Theorem \ref{power}]
Recall $\what\cH_1(\pi)=\{i\in[n]:E_{i}\geq E_{(K)}\}$. For any $i$ such that $E_i\geq n/\alpha$, we have $i\in \what\cH_1(\pi)$. This gives 
\begin{align}
\Pwr_T(\pi)
&=\E\left[\frac{|\cH_1(\pi)\cap\what{\cH}_1(\pi)|}{|\cH_1(\pi)|}\right]
= \frac{1}{|\cH_1(\pi)|}\sum_{i\in \cH_1(\pi)} \Pro \left(i\in\what{\cH}_1(\pi) \right) \\
&\geq \frac{1}{|\cH_1(\pi)|}\sum_{i\in \cH_1(\pi)}\Pro \left(E_{i,T\land\tau_i}\geq\frac{n}{\alpha} \right) \\
&= \frac{1}{|\cH_1(\pi)|}\sum_{i\in \cH_1(\pi)}\Pro \left(\tau_i<\infty \right).
\end{align}
By Condition \ref{powercondition}, we have $\Pro \left(\tau_i<\infty \right)\to 1$ for all $i\in \cH_1(\pi)$, which makes the power tend to unity. This completes the proof. 
\end{proof}

\subsection{Proof of Theorem \ref{lln}}

\begin{proof}[Proof of Theorem \ref{lln}]
Fix arbitrary $i\in\cH_1(\pi)$ throughout the proof. Let $W_{is} = U_{is} + V_{is}$ with
\begin{align}
U_{is} = \log(\lambda_{is}+1-\pi\lambda_{is})-\log(\lambda_i+1-\pi\lambda_i), ~~~
V_{is} = \log(1-\pi\lambda_{is})-\log(1-\pi\lambda_i). 
\end{align}
Condition \ref{cond:2} yields
\begin{align}
\frac{1}{T}\log E_{iT}
&= \frac{1}{T}\sum_{s=1}^{T}X_{is}\log(\lambda_{is}+1-\pi\lambda_{is})+\frac{1}{T}\sum_{s=1}^{T}(1-X_{is})\log(1-\pi\lambda_{is}) \\
&= \bar{X}_{iT}\log(\lambda_{i}+1-\pi\lambda_{i})
+ (1-\bar{X}_{iT})\log(1-\pi\lambda_{i}) + \frac{1}{T}\sum_{s=1}^{T}X_{is}W_{is} \\
&= f(\lambda_i) + o_p(1) + \frac{1}{T}\sum_{s=1}^{T}X_{is}W_{is},
\end{align}
where
\begin{align}
f(x) = \pi_i \log\left(\frac{x+1-\pi x}{1-\pi x}\right) + \log\left(1-\pi x\right).
\end{align}

We will claim $(1/T)\sum_{s=1}^{T}X_{is}W_{is}\xrightarrow{p}0$ ($*$). For any $\epsilon>0$, we have 
\begin{align}
\Pro\left(\left|\frac{1}{T}\sum_{s=1}^{T}X_{is}W_{is}\right|>\epsilon\right)
&\leq \Pro\left(\frac{1}{T}\sum_{s=1}^{T}\left|W_{is}\right| >\epsilon\right)\\
&\leq \Pro\left(\frac{1}{T}\sum_{s=1}^{T}\left|U_{is}\right|>\frac{\epsilon}{2}\right) 
+ \Pro\left(\frac{1}{T}\sum_{s=1}^{T}\left|V_{is}\right|>\frac{\epsilon}{2}\right)\\
&\leq\frac{2}{\epsilon T}\sum_{s=1}^{T}\E[|U_{is}|]
+\frac{2}{\epsilon T}\sum_{s=1}^{T}\E[|V_{is}|].
\end{align}
By Condition \ref{cond:4} and the continuous mapping theorem, we have $|U_{is}|\xrightarrow{p}0$ and $|V_{is}|\xrightarrow{p}0$. 
Thus Lemma \ref{lem:UI} below entails $\E[|U_{is}|]\to0$ and $\E[|V_{is}|]\to0$ as $s\to\infty$. Therefore, the property of Ces\`aro mean implies ($*$). 
In consequence, we have
\begin{align}
\frac{1}{T}\log E_{iT} = f(\lambda_i) + o_p(1).
\end{align}

Note that $x\mapsto f(x)$ is a monotonically increasing and continuous function in $x\in[0,\tilde{\lambda}]$ with $f(0)=0$ and $f(\tilde{\lambda})=\max_{x\geq 0}f(x)$. Thus $f(\lambda_i)>0$ uniformly in $\lambda_i\in(0,\lambda_i^*]$.  This together with Condition \ref{cond:3} implies Condition \ref{powercondition}:
\begin{align}
    \Pro(\tau_i<\infty)
    &\geq \Pro\left(\frac{1}{T}\log E_{iT}\geq\frac{1}{T}\log\frac{n}{\alpha}\right) \\
    &\geq \Pro\left(\left|\frac{1}{T}\log E_{iT}-f(\lambda_i)\right|\leq f(\lambda_i)-O\left(\frac{\log n}{T}\right)\right)\\
    &\to1.
\end{align}
This completes the proof.
\end{proof}

\begin{lem}\label{lem:UI}
If Condition \ref{cond:4} is satisfied, then for any $i\in\cH_1(\pi)$ with given $\pi\in(0,1)$, $\{|U_{is}|\}_{s}$ and $\{|V_{is}|\}_{s}$ are uniformly bounded.
\end{lem}
\begin{proof}[Proof of Lemma \ref{lem:UI}]
Fix arbitrary $i\in\cH_1(\pi)$ throughout the proof. 
%
From the proof of Theorem \ref{lln}, we have
\begin{align}
|U_{is}| 
\leq |\log(\lambda_{is}+1-\pi\lambda_{is})| + |\log(\lambda_i+1-\pi\lambda_i)|.
\end{align}
The first term diverges only if $\lambda_{is}+1-\pi\lambda_{is}\to 0$ (i.e., $\lambda_{is}\to -1/(1-\pi)<0$) or $\to \infty$ in probability, but either does not occur since $\lambda_{it}\in[0,\bar{\lambda}]$ by Condition \ref{cond:4}. The same is said to the second term; in fact, $\lambda_{i}\in(0,\lambda_i^*]$ with $\lambda_i^*=(\pi_i-\pi)/\{\pi(1-\pi)\}<1/\pi$ by Condition \ref{cond:4}. Thus, $|U_{is}|$ is uniformly bounded by some constant. Next, we have
\begin{align}
|V_{is}| \leq |\log(1-\pi\lambda_{is})| + |\log(1-\pi\lambda_i)|. 
\end{align}
This diverges only if either of the following occurs: $\lambda_{is}\to 1/\pi$, $\lambda_{is}\to \infty$, $\lambda_{i}\to 1/\pi$, or $\lambda_{i}\to \infty$ in probability. However, neither happens under Condition \ref{cond:4}. To see this, we note that $\lambda_{is}\leq \bar{\lambda}<1/\pi$ for all $s$ and also $\lambda_i\leq \bar{\lambda}$, which holds by $\lambda_{is}\to \lambda_i$ in probability. Thus, $|V_{is}|$ is also uniformly bounded by some constant. This completes the proof. 
\end{proof}

\subsection{Proof of Theorem \ref{pwr2}}

\begin{proof}[Proof of Theorem \ref{pwr2}]
We will show that $\Pwr_T(\pi)\to 1$ for the BH with $P_{i}^{\normalfont{\textsf{m}}}$'s. First, by the definition of the power, note that 
\begin{align}
\Pwr_T(\pi)
&=\frac{1}{|\cH_1(\pi)|}\E\left[\sum_{i\in\cH_1(\pi)}\mathbb{I}\{i\in\what{\cH}_1(\pi)\}\right]\\
&=\frac{1}{|\cH_1(\pi)|}\sum_{i\in\cH_1(\pi)}\Pro(i\in\what{\cH}_1(\pi))
\geq\frac{1}{|\cH_1(\pi)|}\sum_{i\in\cH_1(\pi)}\Pro\left(P_{i}^{\normalfont{\textsf{m}}}\leq\frac{\alpha}{n}\right).
\end{align}
Therefore, it is sufficient to show that $\lim_{n,T\to\infty}\Pro(P_{i}^{\normalfont{\textsf{m}}}\leq\alpha/n)=1$ for any $i\in\cH_1(\pi)$.

For $S\sim Bin(T,\pi)$ with the probability function, 
\begin{align}
f_S(t) = \binom{T}{t}\pi^t(1-\pi)^{T-t},
\end{align}
the upper $\alpha/n$-quantile is defined as 
\begin{align}
q_{nT}=\min \left\{x\in\mathbb{R}:\Pro(S\geq x)\leq\frac{\alpha}{n} \right\}
= \min \left\{x\in\mathbb{R}:\sum_{t=x}^T f_S(t)\leq\frac{\alpha}{n} \right\}. \label{q}
\end{align}
Corollary 1 in \cite{short2023binomial} with the condition $1-\alpha/n>1/2$ and a standard estimate of the Gaussian quantile \citep{fung2018quantile} yield
\begin{align}
\frac{q_{nT}}{T} 
&\leq \frac{1}{T}\left\lceil T\pi+c_U\sqrt{T}\Phi^{-1}\left(1-\frac{\alpha}{n}\right)+c_U\left\{\Phi^{-1}\left(1-\frac{\alpha}{n}\right)\right\}^2\right\rceil \\
&\leq \pi+\frac{c_U}{\sqrt{T}}\Phi^{-1}\left(1-\frac{\alpha}{n}\right)+\frac{c_U}{T}\left\{\Phi^{-1}\left(1-\frac{\alpha}{n}\right)\right\}^2+\frac{1}{T} \\
&= \pi + O\left(\sqrt{\frac{\log n}{T}}\right)\label{q/T}
\end{align}
for some constant $c_U>0$. 
Thus, we obtain
\begin{align}
\Pro\left(P_{i}^{\normalfont{\textsf{m}}}\leq\frac{\alpha}{n}\right)
&=\Pro\left(\sum_{t=S_i}^T f_S(t)\leq\frac{\alpha}{n}\right) 
=\Pro\left(S_i\geq q_{nT}\right)
=\Pro\left(\bar{X}_{iT}\geq\frac{q_{nT}}{T}\right) \\
&\geq\Pro\left(\bar{X}_{iT}\geq\pi +  O\left(\sqrt{\frac{\log n}{T}}\right)\right) \\
&\geq \Pro\left(|\bar{X}_{iT}-\pi_i| \leq \pi_i - \pi -  O\left(\sqrt{\frac{\log n}{T}}\right)\right), \label{prob_Pm}
\end{align}
where the second equality follows by \eqref{q}. 
For any $i\in\cH_1(\pi)$, the last probability tends to one since the law of large numbers holds by the assumption and $\pi_i - \pi\geq \delta$ for some $\delta>0$. 

Next we consider the power for the BY (i.e., BH with $\{P_{i}^{\normalfont{\textsf{m}}}h_n\}$). In this case, it is sufficient to prove 
$\Pro\left(P_{i}^{\normalfont{\textsf{m}}}\leq \alpha/(nh_n)\right)\to 1$ for all $i\in\cH_1(\pi)$. This only changes the proof to consider $\alpha/(nh_n)$-quantile, denoted by $q_{nT}'$, instead of $\alpha/n$-quantile $q_{nT}$. Then we have
\begin{align}
\frac{q_{nT}'}{T} 
\leq \frac{1}{T}\left\lceil T\pi+c_U\sqrt{T}\Phi^{-1}\left(1-\frac{\alpha}{nh_n}\right)+c_U\left\{\Phi^{-1}\left(1-\frac{\alpha}{nh_n}\right)\right\}^2\right\rceil, 
\end{align}
but, recalling $h_n=\sum_{k=1}^n1/k=\log n +O(1)$, we find this upper bound is the same as \eqref{q/T} since $\Phi^{-1}(1-\alpha/(nh_n))=O(\sqrt{\log(n\log n)})=O(\sqrt{\log n})$. 
Therefore, we obtain the same asymptotic result as \eqref{prob_Pm}. This completes the proof. 
\end{proof}

\subsection{Proof of Proposition \ref{prop:lambda}}

\begin{proof}[Proof of Proposition \ref{prop:lambda}]
Recall $\lambda_{it}=\min\{0\vee \hat{\lambda}_{it},\bar{\lambda}\}$, where $\hat{\lambda}_{it}=(\bar{X}_{i,t-1}-\pi)/(\pi-\pi^2)$. 
Clearly, we have $\lambda_{it}\in[0,\bar{\lambda}]$ a.s. Thus to complete the proof, it suffices to show that $\lambda_{it}\xrightarrow{p}\min\{\lambda_i^*,\bar{\lambda}\}$, where $\lambda_i^*=(\pi_i-\pi)/(\pi-\pi^2)$.   For any $\epsilon>0$, we have 
\begin{align}
\Pro&\left(|\lambda_{it}-\min\left\{\bar{\lambda},\lambda_i^*\right\}|>\epsilon\right)
=\Pro\left(\left|\min\left\{\bar{\lambda}-\lambda_{it},\lambda_i^*-\lambda_{it}\right\}\right|>\epsilon\right)\\
&\leq\Pro\left(\left|\bar{\lambda}-\min\{0\vee \hat{\lambda}_{it},\bar{\lambda}\}\right|>\epsilon\right)
+\Pro\left(\left|\lambda_i^*-\min\{0\vee \hat{\lambda}_{it},\bar{\lambda}\}\right|>\epsilon\right)\\
&=\Pro\left(\left|\min\left\{-\bar{\lambda}\lor(\hat{\lambda}_{it}-\bar{\lambda}),0\right\}\right|>\epsilon\right)
+\Pro\left(\left|\min\left\{-\lambda_i^*\lor(\hat{\lambda}_{it}-\lambda_i^*),\bar{\lambda}-\lambda_i^*\right\}\right|>\epsilon\right)\\
&\leq\Pro(0>\epsilon)+\Pro\left(\left|\max\left\{-\lambda_i^*,\hat{\lambda}_{it}-\lambda_i^*\right\}\right|>\epsilon\right)\\
&=\Pro\left(\left|-\lambda_i^*\right|>\epsilon,\hat{\lambda}_{it}\leq0\right)+\Pro\left(\left|\hat{\lambda}_{it}-\lambda_i^*\right|>\epsilon,\hat{\lambda}_{it}>0\right)\\
&\leq\Pro\left(\hat{\lambda}_{it}\leq0\right)+\Pro\left(\left|\hat{\lambda}_{it}-\lambda_i^*\right|>\epsilon\right)\\
&\leq\Pro\left(\left|\hat{\lambda}_{it}-\lambda_i^*\right|\geq\lambda_i^*\right)+\Pro\left(\left|\hat{\lambda}_{it}-\lambda_i^*\right|>\epsilon\right)\\
&\to0,
\end{align}
where the convergence follows from Condition \ref{cond:2}. This completes the proof. 
\end{proof}

\bibliographystyle{chicago}
\bibliography{ref}
\end{document}